\documentclass[12pt,a4paper]{article}

\usepackage[utf8]{inputenc}
\usepackage{lmodern}
\usepackage{amssymb}
\usepackage{amsthm}
\usepackage{amsmath,amsfonts}
\usepackage{mathabx}
\usepackage[top=3cm,bottom=3cm,left=2.5cm,right=2.5cm]{geometry}
\usepackage{harvard}
\usepackage[shortlabels]{enumitem}
\usepackage{hyperref}
\usepackage{bbm}
\hypersetup{
    colorlinks=false,
    pdfborder={0 0 0},
}

\usepackage{algorithmic}
\usepackage{float}
\algsetup{linenosize=\scriptsize,linenodelimiter=.,indent=1em}
\floatstyle{ruled}
\newfloat{algorithm}{h}{}
\floatname{algorithm}{Algorithm}

\allowdisplaybreaks[3]

\usepackage[center, small]{caption}
\usepackage{subfig}

\usepackage{tikz}
\usetikzlibrary{patterns}
\usetikzlibrary{arrows,backgrounds,calc,decorations.pathreplacing}
\colorlet{myblue}{blue!80!green}
\colorlet{mybluelight}{myblue!50}
\tikzset{
  > = latex',
  axis/.style    = {very thick},
  aborder/.style = {draw},
  acomp/.style   = {fill=black, fill opacity=0.1},
  rect/.style    = {very thick},
  form/.style    = {font=\scriptsize},
  sm/.style      = {font=\small},
  vsm/.style     = {font=\scriptsize}
}

\renewcommand{\t}{\boldsymbol{t}}
\newcommand{\F}{\boldsymbol{F}}
\newcommand{\T}{\boldsymbol{T}}
\newcommand{\f}{\boldsymbol{f}}
\newcommand{\1}{\boldsymbol{1}}
\newcommand{\E}{\mathbb{E}}

\newcommand{\Eis}{\mathbb{E}_{\t_{-i},s}}

\newcommand{\R}{\mathbb{R}}

\newtheorem{theorem}{Theorem}
\newtheorem*{theorem*}{Theorem}
\newtheorem{lemma}{Lemma}
\newtheorem{definition}{Definition}

\newtheorem{remark}{Remark}
\newtheorem{proposition}{Proposition}
\newtheorem{corollary}{Corollary}
\newtheorem{example}{Example}

\newcommand{\eg}{e.\,g.}

\begin{document}

\title{Optimal allocations with capacity constrained verification\thanks{We thank Maria Goltsman, Thomas Brzustowski, Axel Niemeyer, and Rakesh Vohra for comments and discussions. Erlanson: Department of Economics at University of Essex,  \texttt{albin.erlanson@essex.ac.uk}; Kleiner: Department of Economics, University of Bonn, \texttt{andreas.kleiner@uni-bonn.de}. Kleiner 
acknowledges financial support
from the German Research Foundation (DFG) through Germany’s Excellence Strategy - EXC 2047/1
- 390685813, EXC 2126/1-390838866 and the CRC TR-224 (Project B01).}}
\author{Albin Erlanson \and Andreas Kleiner}
\date{\today\\}

\maketitle

\begin{abstract}
A principal has $m$ identical objects to allocate among a group of $n$ agents. Objects are desirable and the principal's value of assigning an object to an agent is the agent's private information. The principal can verify up to $k$ agents, where  $k<m$, thereby perfectly learning the types of those verified. We find the mechanism that maximizes the principal's expected utility when no monetary transfers are available. In this mechanism, an agent receives an object if (i) his type is above a cutoff and among the $m$ highest types, (ii) his type is above some lower cutoff but among the $k$ highest types, or (iii) he receives an object in a lottery that allocates the remaining objects randomly. \\


\textit{Keywords}:  Mechanism Design with Evidence; Allocations;  Verification
\\

\textit{JEL classification}: D82, D71
\end{abstract}

\newpage

\section{Introduction}\label{sec:intro}


We examine an allocation problem involving multiple identical objects. Agents possess private information about the principal's value of assigning them an object and each agent desires one object but has no use for more than one. There are no monetary transfers and agents' private information is based on verifiable evidence. The principal can check this evidence, but is limited to a fixed number of checks---fewer than the number of available objects. We study how to optimally check agents based on their reports and use the results, along with agents' reports, to allocate objects. 
This problem differs from standard mechanism design due to the verifiability of private information, which 
arises in various economic contexts.

A possible application of our model is the allocation of work permits to skilled immigrants. The value of assigning a permit to an applicant is a function of her characteristics (e.g., skills, education, work experience, language ability). Therefore, we are in a setting with private but verifiable information.\footnote{Canada has a program for attracting skilled immigrants that resembles the model we build. The score an applicant gets is based on ``\textit{factors known to contribute to economic success}'', for details see:  https://www.canada.ca/en/immigration-refugees-citizenship/services/immigrate-canada/express-entry/works.html.} The objective is to maximize the total value from assigning the permits. 
As is the case in the example with work permits in Canada, we do not allow for the possibility of selling work permits. We analyze the optimal method to allocate $m$ permits when at most $k<m$ agents can be checked. 

To do this, we formulate a model of capacity-constrained verification in an allocation problem without transfers. This model captures the situation when there is a fixed monetary budget for checking or a given team of experts available for reviewing reports. Before deciding on who receives an object, the principal can check some of the agents and learn their information. Strictly fewer agents can be checked than there are objects available to allocate. The mechanism design question is whom to check based on the submitted reports and how to allocate objects based on the reports and checks. If agent $ i$ is checked, the principal learns his type $t_i$ perfectly. The principal gets value $t_i$ if agent $i$ is assigned an object, and the total payoff to the principal is equal to the sum of the types of all agents that get an object. There is no additional cost of checking an agent, the only constraint is that at most $k$ agents can be checked.

\medskip

We find the mechanism that maximizes the expected payoff from allocating $m$ objects with $k$ checks. 
The allocation rule in the optimal mechanism can be decomposed into two rules, which can be thought of as being used in two stages: a merit-based rule in the first stage,  and a lottery-based rule in the second stage. The merit-based rule allocates a number of objects ``efficiently'' to the agents with the highest reported types, where the number of objects that are allocated efficiently is typically determined by two thresholds as follows: If less than $k$ agents report a type above the lower threshold, then all these agents obtain an object; the remaining objects enter the second stage. If more than $k$ agents report a type above the lower threshold, such an agent gets an object if either (i) her report is among the $k$ highest ones, or (ii) her report is among the $m$ highest ones and her report is above the higher threshold. Among agents that get an object in the first stage, $k$ are randomly selected and their reports are checked; anyone found lying will not get an object in any case. Any agents that didn't get an object in the first stage (and weren't caught lying) enter a second stage, where a lottery-based rule allocates the remaining objects randomly in a way that all types below the lower threshold receive an object with the same probability.

To explain the optimal mechanism in more detail, we discuss the incentives of the agents to be truthful.
Because the merit-based rule sometimes allocates more than $k$ objects, not all agents that get an object in the first stage can be checked. This suggests that agents might be tempted to claim that they have a high type. The optimal mechanism provides incentives to be truthful by not always allocating all objects in the first stage; the lottery-based allocation of the remaining objects in the second stage ensures that even agents with low types have a chance to win and are not tempted to lie about their type: their expected probability of getting an object in the lottery equals the expected probability of getting an object in the first stage by lying. However, if a low type knew that more than $m$ agents have a type above the higher type, she would anticipate that all objects will be allocated in the first stage and she won't get an object in the second stage by being truthful. In that case, she would have a strict incentive to claim a high type, hoping to get an object in the first stage without being checked. This shows that the optimal mechanism, while being Bayesian incentive compatible, does not satisfy stronger ex-post incentive constraints. This contrasts with models of costly verification, where optimal mechanisms can typically be implemented to satisfy these stronger incentive constraints.

Methodologically, we formulate the principal's problem using ``interim rules.'' Since there is no point in checking an agent that won't get an object, we obtain type-dependent feasibility constraints concerning which agents can be checked. This requires a characterization of which interim rules are feasible in the presence of type-dependent capacity constraints.
  
\bigskip
\subsection{Related Literature}\label{subsec:rel_lit}
The literature on costly state verification in principal-agent models was initiated by  \citeasnoun{townsend79}.  Our model differs from his, and the literature building on it \citeaffixed{galeHellwig85,borderSobel87}{see \eg}, since monetary transfers are not feasible in our model.  In models without transfers, incentive constraints and economic trade-offs are different than in models with money.

Over the last years, there has been growing interest in models with state verification that do not allow for transfers,  as initiated by \citeasnoun{ben-porath14}.  The model studied by \citename{ben-porath14} consists of a principal that wishes to allocate a single indivisible good among a group of agents, and each agent's type can be learned at a given cost. They find that the expected utility-maximizing mechanism for the principal is a favored-agent mechanism, where a pre-determined favored agent receives the object unless another agent claims a value above a threshold, in which case the agent with the highest (net) type gets the object. 
Our model is similar in that there are no monetary transfers but differs in that our principal allocates several objects and can check a fixed number of agents at no cost. Compared to a model with a possibly unlimited number of tests at a constant marginal cost, this leads to a richer set of trade-offs that determine optimal mechanisms. 

\citeasnoun{ben-porath17} analyze optimal mechanisms in a model of Dye evidence and apply their results to a model in which the principal allocates $m$ objects and can verify agents at a constant marginal cost. They show that the optimal mechanism generalizes the favored-agent mechanism in a natural way to an $m$-favored-agents mechanism. Their approach implies that the optimal mechanism satisfies stronger ex-post incentive constraints, which shows a conceptual difference to our model, where the optimal Bayesian incentive compatible mechanism is not ex-post incentive compatible. This also suggests that the optimal mechanism in our model cannot be obtained using their general results for the Dye-evidence model.
Other related papers on mechanism design with evidence include \citeasnoun{glazerRubinstein04} and \citeasnoun{glazerRubinstein06}, which consider a situation in which an agent has private information about several characteristics and tries to persuade a principal to take a given action, and the principal can only check one of the agent's characteristics.\footnote{For additional papers on mechanism design with evidence, see also \citeasnoun{greenLaffont86}, \citeasnoun{bullWatson07}, \citeasnoun{deneckereSeverinov08}, \citeasnoun{benporathLipman12}, and \citeasnoun{schweighofer2024}.} 
Other papers on optimal mechanisms in settings with evidence include \citeasnoun{mylovanov17}, \citeasnoun{halac2020}, and \citeasnoun{Kattwinkel2023}.

Methodologically, we use interim rules to find the optimal mechanism. This approach was pioneered by~\citeasnoun{maskinRiley84}, who used it to characterize optimal auctions with risk-averse bidders. This methodology has later been used repeatedly in auction theory and mechanism design successfully. A key question asks for which interim rules there is a feasible ex-post rule that induces the interim rule as its marginal distribution. This question was first answered conclusively by~\citeasnoun{border91}, who characterized the set of feasible interim allocation rules.
More recently, Border's result was generalized by~\citeasnoun{che13}, who among other things, allow for the allocation of $m$ objects when agents face capacity constraints. We cannot use their result directly since we need to incorporate type-dependent capacity constraints as part of the feasibility requirements on interim rules. However, we show that one can use the same proof technique as in \citeasnoun{che13} to obtain a characterization of feasible interim rules with type-dependent ex-post constraints. For finite type spaces, a closely related result was established independently in \citeasnoun{valenzuela2022greedy}.

\bigskip

The rest of the paper is organized as follows. In Section 2, we present the model and the principal's objective, and we also characterize the class of incentive-compatible mechanisms. In Section 3, we derive the optimal mechanism, and we discuss how to optimize within the class of optimal mechanisms.  Section 4  contains a proof sketch, and Section 5 concludes the paper.

\section{Model and preliminaries}\label{sec:Model}

A principal allocates $m$ identical and indivisible objects to a set $N:=\{1,\dots,n\}$ of agents, where $m<n$. Each agent $i$ is characterized by a type $t\in T:=[0,1]$, which is private information to agent $i$. Types are drawn independently according to a strictly positive density $f$ with cdf $F$ for each agent $i$.  Let $\t = (t_i)_{i\in N}$, $\t_{-i} = (t_j)_{j\in N\setminus\{i\}}$, $\T= [0,1]^N$, $\F(\t)= \prod_{i\in N} F(t_i)$, and $\f(\t)= \prod_{i\in N} f(t_i)$. Agent $i$  with type $t$ gets utility $u_i(t)>0$ if he is assigned at least one object and zero otherwise. No agent needs more than one object.  The principal's payoff equals the sum of the types of all agents who obtain an object.
The principal has a verification technology at his disposal and can learn the types of up to $k<m$ agents perfectly (we will use verification and audit interchangeably). The principal can commit to arbitrary mechanisms to allocate the objects, and we consider Bayes-Nash equilibria of the induced game. 

We now define direct mechanisms for this setting. We capture stochastic mechanisms by using a random variable $s$, which is distributed independently of the types and uniformly on $[0,1]$ and whose realization is only observed by the principal.   
Formally, an \textit{allocation rule} is a profile of functions $p_i:\T\times[0,1]\rightarrow \{0,1\} $ satisfying, for all $\t$ and 
  $s$, $\sum_{i=1}^n p_i(\t,s) \le m$ and an \textit{audit rule} is a profile of functions $a_i:\T\times[0,1]\rightarrow \{0,1\}$  satisfying, for all $\t$ and  $s$, $\sum_{i=1}^n a_i(\t,s) \le k$. 
A \textit{direct mechanism} $(p,a)$ consists of an allocation rule and an audit rule.
Given a profile of reported types $\t$ and a realization $s$ of the random variable, agent $i$ is audited if $a_i(\t,s)=1$ and obtains an object if all audited agents were found to be truthful and $p_i(\t,s)=1$ and obtain no object otherwise. One can show that it is without loss of generality to restrict attention to direct mechanisms that are Bayesian incentive-compatible.\footnote{This can be shown using arguments similar to those in \citeasnoun{ben-porath14} and \citeasnoun{erlansonKleiner18}.}

\begin{definition}\label{def:BIC}
 A mechanism $(p,a)$ is Bayesian incentive compatible (BIC) if, for all $i\in N$ and all $t,t'\in T$
  \begin{equation}
     u_i(t')\cdot \E_{\t_{-i},s}[p_i(t',\t_{-i},s)] \geq 
       u_i(t') \cdot  \E_{\t_{-i},s}[p_i(t,\t_{-i},s) (1-a_i(t,\t_{-i},s))].\label{eq:BIC_def}   
  \end{equation}
 \end{definition}

The principal's problem can formally be stated as:

\begin{align*} 
	\underset{p,a}\max \ &\E_{\t,s}  \left[\sum_{i\in N} p_i(t_i, \t_{-i},s)t_i \right] \\
	& \text{s.t. $(p,a)$ is BIC.}
\end{align*}

\paragraph{Reformulating the principal's problem}
Due to the symmetry of this problem, it is without loss of optimality to restrict attention to symmetric mechanisms that treat all agents identically, and henceforth we will do that.\footnote{That is, we focus in the following on mechanisms satisfying, for all $\t$ and permutations $\sigma$, $\E_s[p_i(\t,s)]=\E_s[p_{\sigma(i)}(\t_{\sigma},s)]$, and similarly for $a$ and $pa$. }

It is useful to formulate this mechanism design problem  in terms of interim rules. Given a mechanism $(p,a)$, its  \emph{interim rules} are defined for all $t\in T$ by
\begin{align*}
    P(t)&:= \E_{\t_{-i},s}[p_i(t,\t_{-i},s)] \\
    A(t)&:= \E_{\t_{-i},s}[a_i(t,\t_{-i},s)].
\end{align*}
We say that $p$ \emph{induces} $P$ and $a$ \emph{induces} $A$. 

We begin by simplifying the Bayesian incentive compatibility constraints.
First, since $u_i(t')>0 $ for all $i\in N$ and all $t'\in T$, $u_i(t')$ cancels from both sides of equation~\eqref{eq:BIC_def} above. Thus, a mechanism is BIC if and only if, for all $t,t'\in T$,
\begin{align}\label{eq:BIC}
\E_{\t_{-i},s}[p_i(t',\t_{-i},s)]\ge \E_{\t_{-i},s}[p_i(t,\t_{-i},s) (1-a_i(t,\t_{-i},s))]
\end{align}

We want to formulate this inequality only in terms of interim rules but cannot do so directly (for example,
$\E_{\t_{-i},s}[p_i(t,\t_{-i},s)a_i(t,\t_{-i},s)]\neq P(t)A(t)$ in general).
However, note that once the random variable $s$ is drawn either $p_i(\t,s)$ is one or zero. If $p_i(t,\t_{-i},s)=0$ then there is no need to audit agent $i$ at $(t,\t_{-i},s)$ because he cannot be penalized if he lied. Hence to achieve truth-telling only agents that are scheduled to get an object need to be audited. That is,  it is without loss of generality to assume that $a_i(\t,s)=1$ only if $p_i(\t,s)=1$, and we will from now on restrict attention to mechanisms satisfying this property. Using this observation, 
$\E_{\t_{-i},s}[p_i(t,\t_{-i},s) a_i(t,\t_{-i},s)]=\E_{\t_{-i},s}[ a_i(t,\t_{-i},s)]$ and the right-hand side of \eqref{eq:BIC} simplifies to $P(t)-A(t)$.

These observations imply the following characterization of Bayesian incentive compatibility.

\begin{lemma}\label{lemma:bic_simple}
A symmetric mechanism $(p,a)$ is BIC if and only if its interim rules $(P,A)$ satisfy, for all $t\in T$,
\begin{align}\label{eq:bic2}
   A(t) \geq   P(t) - \inf_{t'\in T} P(t').
\end{align}
\end{lemma}

To formulate the principal's problem with interim rules, we introduce the following definition.

\begin{definition}
An audit rule $a$ is \emph{feasible given an allocation rule $p$} if $a_i(\t,s)\le p_i(\t,s)$ for all $i\in N$, $\t\in \T$, and $s\in [0,1]$. 
A pair of interim rules $(P,A)$ is \emph{feasible} if there exists an allocation rule $p$ and an audit rule $a$ which is feasible given $p$ that induce $(P,A)$.
\end{definition}

We can now reformulate the principal's problem as follows:
\begin{align*} \label{P_int}
	\underset{P,A} \max \ & n \E_t[P(t)t] \tag{Opt}\\
	& \text{s.t. $(P,A)$ is feasible and }\\
	& A(t) \ge   P(t) - \inf_{t'\in T} P(t') \text{ for all } t\in T. 
\end{align*}

\section{Optimal mechanism}\label{sec:opt_mech}

In this section we define the class of merit-with-guarantee rules and show that an optimal mechanism can always be found in this class.
\subsection{Finding the optimal interim rules}\label{subsec:opt_reduced_form}

To build some intuition for the optimal mechanism, we start by discussing possible allocation rules that the principal could (or couldn't) use. 
The first best allocation rule allocates objects to the agents with the $m$ highest types. However, if only $k<m$ agents can be verified, this rule cannot be part of an incentive-compatible mechanism: the lowest possible type would receive an object with probability zero by being truthful, but would receive one with positive probability by reporting some higher type since not every agent getting an object can be verified. One possibility to create an incentive-compatible mechanism is to reduce the number of objects that are allocated ``efficiently''; for example, the principal could verify agents sending the $k$ highest reports and allocate an object to the agent if he was truthful. The remaining $m-k$ objects could be allocated randomly among the agents that did not send one of the $k$ highest reports. Note that the incentive constraints are slack in this mechanism: any misreport can only change the allocation if the misreport is among the $k$ highest reports. However, in that case, the report is verified for sure and by misreporting there is a zero probability of obtaining an object, whereas the probability of receiving an object would be strictly positive by being truthful. 

The fact that incentive constraints are slack implies that the principal can improve this mechanism. One possibility to do so is to (sometimes) allocate $k+1$ or more objects efficiently while allocating the remaining ones randomly. Since incentive constraints were previously slack, they will still be satisfied if the expected number of objects that are allocated efficiently is not too large.  In fact, the mechanism can further be improved by introducing a cutoff and allocating to the agents with the $m$ highest reports as long as these reports are above the cutoff; $k$ of these reports are verified, and all remaining objects are allocated randomly among agents with reports below the cutoff. This mechanism is incentive-compatible as long as the cutoff is high enough and we will see that the optimal mechanism shares similar features. We will show that this mechanism can be further improved by introducing a second, lower cutoff: An agent then gets an object if he reports above the high cutoff and his report is among the $m$ highest ones, or if his report lies between the two cutoffs and his report is among the $k$ highest ones. The principal verifies $k$ of those high reports, and all remaining objects are allocated randomly. 

\medskip 

To describe the optimal interim rules, we start by deriving constraints that any feasible mechanism must satisfy. These constraints (in \eqref{eq:c_allo}, \eqref{eq:c_ic}, and \eqref{eq:c_aud}) are due to the limited availability of objects and audits \citeaffixed{border91}{for related conditions, see}. We then construct an interim allocation rule that satisfies the relevant constraints as equalities and show that there is an optimal mechanism with such an interim allocation rule.

A feasibility constraint on $P$, which stems from at most $m$ objects being available, is that for any type $t$,
\begin{equation}\label{eq:c_allo}
 n\int_t^1 P(x)\,\mathrm dF(x)  \le c^{allo}(F(t)) 
\end{equation}
where
\begin{equation}
 c^{allo}(q) := \sum_{i=1}^n\min\{i,m\}\binom{n}{i}(1-q)^{i}q^{n-i}. \nonumber
\end{equation}
To see this, note that the left-hand side denotes the expected number of agents with a type above $t$ that are assigned an object. For any $i\in\{1,\dots ,n\}$, the probability that there are $i$ agents with a type above $t$ is $\binom{n}{i}(1-F(t))^{i}F(t)^{n-i}$, and at most $\min\{i,m\}$ of these agents can be assigned an object. By summing over $i$, we obtain (on the right-hand side) an upper bound on the expected number of agents with a type above $t$ that can be assigned an object, and the inequality follows. 

We have seen in Lemma \autoref{lemma:bic_simple} that whether a mechanism is BIC depends in particular on the lowest probability with which some type obtains an object; for an interim allocation rule $P$, let therefore $\varphi:=\inf_{t'} P(t')$.
We then obtain that for any type $t$,
\begin{equation}\label{eq:c_ic}
 n\int_t^1 P(x)\,\mathrm dF(x)  \le c_{\varphi}^{ic}(F(t)),
\end{equation}
where
\[ c_{\varphi}^{ic}(q):= m - n q \varphi. \]
Similarly to inequality \eqref{eq:c_allo}, the left-hand side in \eqref{eq:c_ic}   denotes the expected number of agents with a type above $t$ receiving an object.
Because there are $m$ objects and since in expectation $nF(t)$ agents have a type below $t$ and each of them receives an object with a probability of at least $\varphi$, the inequality follows.

Since at most $k$ audits are available, we obtain in analogy with inequality \eqref{eq:c_allo} the following feasibility constraint on $A$ that must hold for all types $t$: 
\[n\int_t^1 A(x) \,\mathrm dF(x)  \le \sum_{i=1}^n\min\{i,k\}\binom{n}{i}(1-F(t))^{i}F(t)^{n-i}.\]
 Lemma \ref{lemma:bic_simple} implies that if $(P,A)$ is BIC then $A(t) \ge P(t)-\inf_{t'} P(t')$. Plugging this into the inequality above and setting $\varphi:=\inf_{t'} P(t')$, we obtain that for any type $t$,
\begin{equation}\label{eq:c_aud}
 n\int_t^1 P(x) \,\mathrm dF(x)  \le c_{\varphi}^{aud}(F(t)),
\end{equation} 
where 
\[ c_{\varphi}^{aud}(q) :=\sum_{i=1}^n\min\{i,k\}\binom{n}{i}(1-q)^{i}q^{n-i}+ n (1-q) \varphi. \] 
Finally, we combine the constraints in inequalities \eqref{eq:c_allo},  \eqref{eq:c_ic}, and \eqref{eq:c_aud}  to obtain
\begin{align}\label{eq:upper_bound}
n \int_{t}^1 P(x) \,\mathrm dF(x) \le c_{\varphi}(F(t)):=\min\{c^{allo}(F(t)),c_{\varphi}^{aud}(F(t)),c_{\varphi}^{ic}(F(t))\}.    
\end{align}
One can view inequality \eqref{eq:upper_bound} as being parametrized by $\varphi$; any feasible and BIC pair of interim rules $(P,A)$ with $\inf_{q'} P(q')= \varphi$ must satisfy this parametrized inequality \eqref{eq:upper_bound}.
We define $C^{allo}$ to be the set of types for which $c^{allo}_{\varphi}$ is the binding constraint,
\[
C^{allo}:=\{t\in T: c^{allo}(F(t))\le  c_{\varphi}^{aud}(F(t))\text{ and } c^{allo}(F(t))\le c_{\varphi}^{ic}(F(t))\},
\]
and $C^{aud}$ and $C^{ic}$ are defined  analogously.\footnote{Since $c^{aud}_{\varphi}$ is concave, $c^{ic}_{\varphi}$ is linear, and $c^{aud}_{\varphi}(0)\ge c^{ic}_{\varphi}(0)$ whenever $\varphi\ge \frac{m-k}{n}$, $C^{ic}$ is an interval starting at $0$. It can also be shown that $C^{aud}$ is always an interval and that $C^{allo}$ consists of at most two intervals with an endpoint of 1 (see Lemma \ref{lemma:intervals_new} for the complete description).} 

Figure 1 illustrates by plotting the right-hand sides of \eqref{eq:c_allo}, \eqref{eq:c_ic}, and \eqref{eq:c_aud}. The black curve is the lower envelope and corresponds to the right-hand side of \eqref{eq:upper_bound}.

\begin{figure}[ht!]
\includegraphics[width=\linewidth]{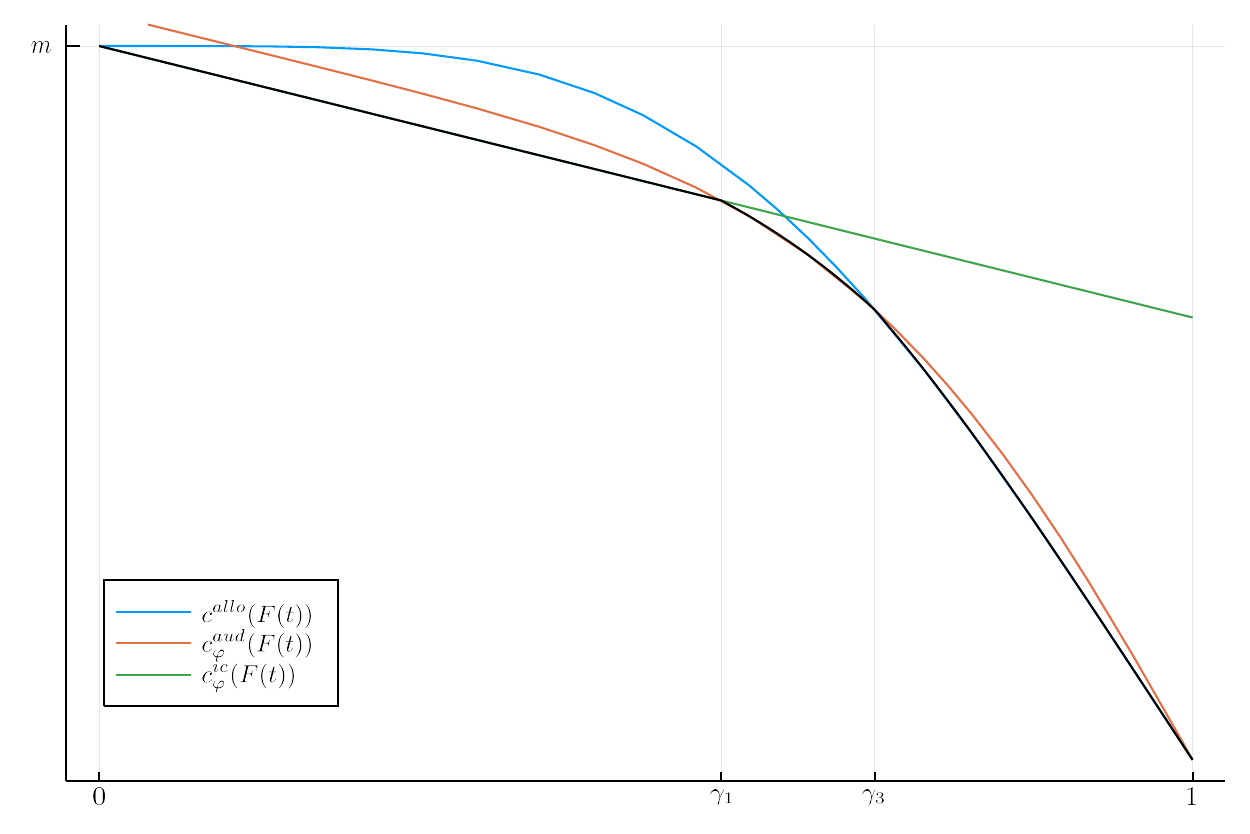}
\caption{Necessary constraints}
\label{fig:const}
\end{figure}

For any $0\le \varphi\le \frac{m}{n}$, we define a pair $(P,A)$ of interim rules that satisfies \eqref{eq:upper_bound} as an equality and therefore, in a suitable sense, allocates as much as possible to high types.

\begin{definition}\label{reduced:opt_mech}
Given  $0\le \varphi\leq \frac{m}{n}$, the \emph{merit-with-guarantee rule} is given by the interim  allocation rule\footnote{At points where $c$ is not differentiable, we use the right-derivative.} 
\begin{equation}\label{eq:reduced_opt_allocation}
 P(t):=-\frac{1}{n} c_{\varphi}'(F(t)),
\end{equation}
and the interim audit rule
\begin{equation}\label{eq:reduced:verif_opt_audit}
A(t):= P(t)-\varphi.
\end{equation} 
\end{definition}

We will show below that this interim allocation rule can be implemented as follows: An agent reporting a type in $C^{allo}$ obtains an object if and only if his report is among the $m$ highest ones (presuming he is not caught lying), which corresponds to the efficient allocation rule. An agent reporting a type in $C^{aud}$ obtains an object if his report is among the $k$ highest ones (again, presuming he is not caught lying). All remaining agents and objects enter a lottery, which ensures that each type has an interim probability of getting an object that is at least $\varphi$. We will see that by selecting $k$ of the high reports for an audit, such an allocation rule is part of an incentive-compatible mechanism.

We can now state our main result. 

\begin{theorem}\label{thm:opt}
A merit-with-guarantee rule maximizes the principal's expected utility.
\end{theorem}

We provide a proof sketch in Section \ref{sec:proof_red_optimal} and relegate all technical details to the appendix.
To obtain intuition for the result, we argue below that no feasible mechanism can allocate more objects to high types than the allocation rule of some merit-with-guarantee rule. The main argument in the proof of Theorem \ref{thm:opt} is then to establish that there indeed exists a feasible BIC mechanism $(p,a)$ corresponding to a merit-with-guarantee rule.

\medskip

\begin{figure}[ht!]
\includegraphics[width=\linewidth]{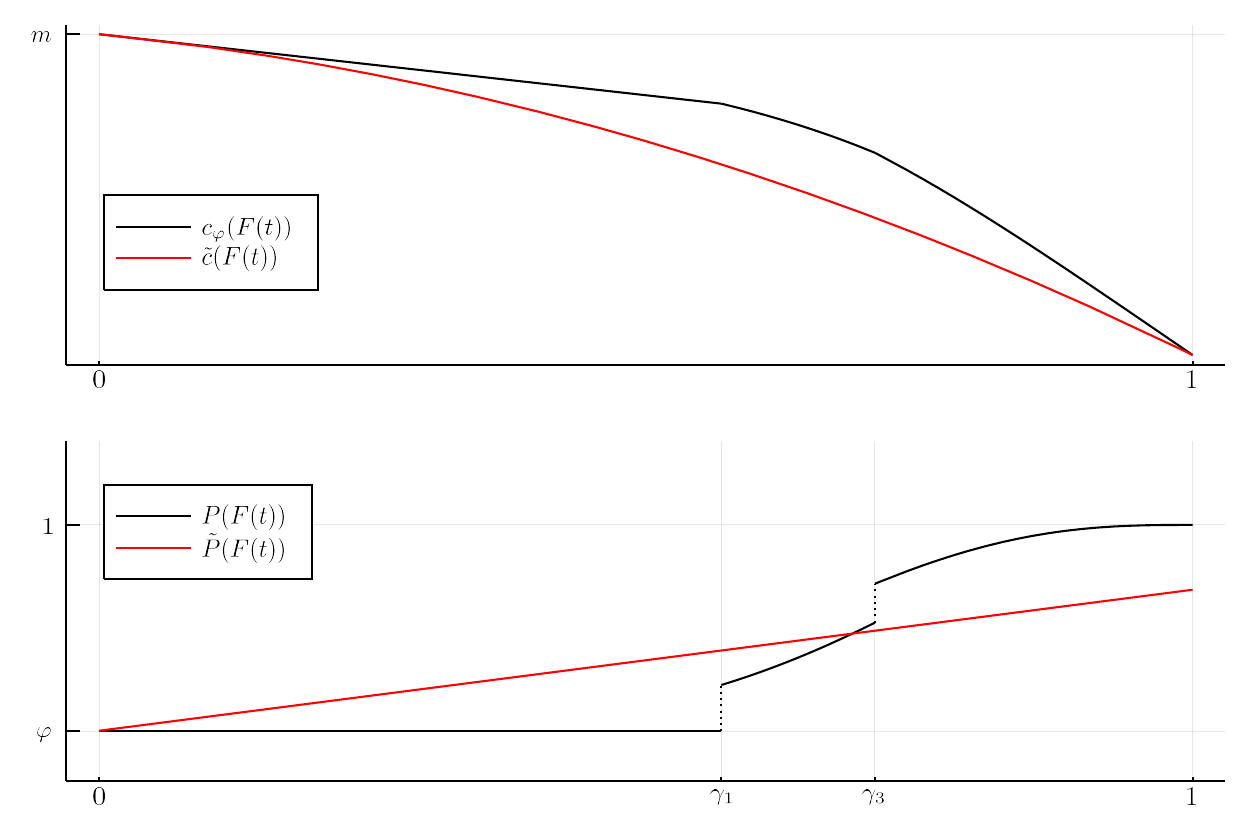}
\caption{Illustrating the optimal mechanism }
\label{fig:opt_mech}
\end{figure}

\medskip

To obtain insights into why a merit-with-guarantee rule is optimal for the principal, consider an alternative feasible mechanism with interim allocation rule $\tilde P$ that satisfies $\inf \tilde P\ge \varphi$. Then the probability that an agent with type above $t$ obtains an object under $\tilde P$ is at most the right-hand side of \eqref{eq:upper_bound}. In the upper panel of Figure \ref{fig:opt_mech}, the lower curve, denoted by $\tilde{c}$, shows the left-hand side of \eqref{eq:upper_bound} for the alternative interim allocation rule $\tilde P$. Note that the derivative of this curve (multiplied at each $t$ by $-\frac{1}{nf(t)}$) is the interim probability with which a given agent with type $t$ receives an object. This interim rule is shown in the lower panel alongside the interim allocation rule for the merit-with-guarantee rule. It follows that the interim allocation rule of any feasible mechanism satisfying $\inf \tilde P\ge \varphi$ will be less likely to allocate to a high type than the interim allocation rule of a merit-with-guarantee rule. This implies that a merit-with-guarantee rule is optimal if one can always construct allocation and audit rules implementing such rules.

\medskip
The class of merit-with-guarantee rules is parameterized by the one-dimensional parameter $\varphi\in[0,\frac{m}{n}]$, which determines the probability with which the lowest type obtains an object. Which values of $\varphi$ can be optimal? First, since there are only $m$ objects, $\varphi$ can be at most $\frac{m}{n}$. Second, for all values $\varphi< \frac{m-k}{n}$ the constraint  $c^{ic}_{\varphi}$ lies  always above $c^{aud}_{\varphi}$. This can be  seen in  Figure \ref{fig:const} by noting that for any  $\varphi<\frac{m-k}{n}$ we have that $c^{ic}_{\varphi}(0)> c^{aud}_{\varphi}(0)$, and that $c^{aud}_{\varphi}(\cdot)$ decreases faster than $c^{ic}_{\varphi}(\cdot)$. Note further that $\frac{\partial}{\partial \varphi}c^{aud}_{\varphi}(q)> 0$, which implies that a higher value of $\varphi$ relaxes the constraint $c^{aud}_{\varphi}$. Thus, for any parameter value of $\varphi< \frac{m-k}{n}$, by increasing $ \varphi$  we obtain another merit-with-guarantee rule that dominates the previous one in a first-order stochastic dominance sense. Therefore, the optimal choice of $\varphi$ is at least $\frac{m-k}{n}$.     
In Subsection \ref{subsec:foc}, we provide a first-order condition for the optimal choice of $\varphi$. 

By restricting attention to values of $\varphi$ that can be optimal we obtain the following explicit description of when each constraint in \eqref{eq:upper_bound} is binding.  

\begin{corollary}\label{lemma:intervals}
 Suppose $\varphi\in[\frac{m-k}{n},\frac{m}{n}]$. Then there are numbers $\gamma_1, \gamma_2, \gamma_3$ such that $0\le \gamma_1 \le \gamma_2\le \gamma_3\le1$ and the following conditions hold:
 \begin{align*}
 c^{allo}(F(t))&\le c_\varphi^{aud}(F(t)) \text{ and } c^{allo}(F(t)) \le c_\varphi^{ic}(F(t)) & &\text{ for all } t\in [\gamma_1,\gamma_2)\cup [\gamma_3,1] \\
 c_\varphi^{aud}(F(t))&\le c^{allo}(F(t)) \text{ and } c^{aud}_\varphi(F(t)) \le c^{ic}_\varphi(F(t))  & &\text{ for all } t\in [\gamma_2,\gamma_3]\\
 c_\varphi^{ic}(F(t))  &\le c^{allo}(F(t)) \text{ and } c_\varphi^{ic}(F(t)) \le c^{aud}_\varphi(F(t)) & &\text{ for all } t\in [0,\gamma_1].
 \end{align*}
 \end{corollary}

Figure \ref{fig:opt_mech_cutoffs}  illustrates how the \emph{merit-with-guarantee rule} allocates depending on which upper bound in \eqref{eq:upper_bound} is binding. The type space is partitioned into three regions $C^{allo}$, $C^{aud}$, and $C^{ic}$, and on each region a particular allocation rule is used. Corollary \ref{lemma:intervals} establishes that each region is an interval (except possibly $C^{allo}$, which can be a union of two intervals).\footnote{See Lemma \ref{lemma:intervals_new} in Appendix \ref{App:prop_c}  for the complete characterization of which constraint is binding, for all feasible values of $\varphi$.} 

\begin{figure}[ht!]
\begin{center}
\begin{tikzpicture}[decoration={brace},baseline,x=4.5cm,y=3.5cm]
      \draw[axis] (0,0) -- (3.4,0);
      \draw[thick] (0,-0.05) node[below] {$0$} -- (0,0.05);
      \draw[thick] (3.4,-0.05) node[below] {1} -- (3.4,0.05);
      \draw[thick] (0.8,-0.05) -- (0.8,0.05);
       \draw[thick] (1.4,-0.05) -- (1.4,0.05);
      \draw[thick] (2.4,-0.05) -- (2.4,0.05);
      \draw (0.8,0.05)node[above]{$\gamma_1$} (1.4,0.05)node[above]{$\gamma_2$} (2.4,0.05)node[above]{$\gamma_3$} ;
      \draw (0.5,-0.05) node[below] {$C^{ic}$} (1.1,-0.05) node[below] {$C^{allo}$} (1.8,-0.05) node[below] {$C^{aud}$} (2.8,-0.05) node[below] {$C^{allo}$};     
\end{tikzpicture} 
\end{center}
\caption{Illustrating optimal mechanism: cutoffs }
\label{fig:opt_mech_cutoffs}
\end{figure}
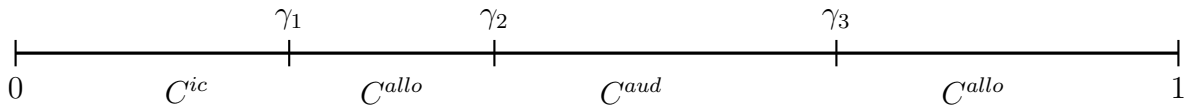

\subsection{How to induce the optimal allocation rule}\label{subsec:ex-post_rule}

In this subsection, we describe an allocation rule that induces the optimal interim allocation rule $P$ described in Definition \ref{reduced:opt_mech}; this allocation rule consists of a merit-based allocation and a lottery-based allocation.

\begin{definition}\label{def:ex_post_tilde}
Given $C^{allo}$ and $C^{aud}$, the \emph{merit-based allocation} $p^m:\T\rightarrow \mathbb{R}^n$ is defined as follows: $p^m_i(\t)=1 $ if $t_j\neq t_k$ for all $j\neq k$, and either $(i)$  $t_i\in C^{allo}$ and $t_i$ is among the $m$ highest reports, or $(ii)$   $t_i\in C^{aud}$ and $t_i$ is among the $k$ highest reports; otherwise $p^m_i(\t)=0$.\footnote{For simplicity, we specify that no objects are allocated in case of a tie, which yields a symmetric and deterministic allocation rule. Because ties have probability 0, this does not affect interim rules and expected payoffs.}
\end{definition}

To illustrate the merit-based allocation, consider the interim allocation rule shown in Figure \ref{fig:opt_mech}.   There are two cutoffs, $\gamma_1$ and $\gamma_3$, such that $C^{allo}=[\gamma_3,1]$, $C^{aud}=[\gamma_1,\gamma_3]$, and $C^{ic}=[0,\gamma_1]$.\footnote{In the example above note that $\gamma_1=\gamma_2$.} Accordingly, the merit-based allocation assigns an object to every agent whose type is among the $m$ highest ones and above $\gamma_3$ and to every agent whose type is among the $k$ highest ones and between $\gamma_1$ and $\gamma_3$.

We now argue that the merit-based allocation induces $P$ on $C^{allo}$: By allocating all objects to the $m$ agents with the highest types as long as these lie in $C^{allo}$,
the expected number of agents with a type $t\in [\gamma_3,1]$ who obtain an object is
\[ \sum_{i=1}^n \min\{i,m\}\binom{n}{i}(1-F(t))^{i}F(t)^{n-i},
\]
which equals $c^{allo}(F(t))$. Therefore, inequality \eqref{eq:c_allo} is binding for all $t\ge \gamma_3$ for the interim allocation rule induced by merit-based allocation. We conclude that the merit-based allocation induces $P$ on $C^{allo}$.

One can show along similar lines that the merit-based allocation induces $P-\varphi$ on $C^{aud}$.
Observe that a merit-based allocation does not allocate all $m$ objects if less than $m$ agents have a type in $C^{allo}$. To find an allocation rule that induces $P$ for all types, we need to supplement the merit-based allocation with a second-stage allocation that allocates the remaining objects to agents that did not get an object under the merit-based allocation and such that each agent with a type in $C^{aud}\cup C^{ic}$ gets an object in the second-stage allocation with probability $\varphi$. We say an allocation rule is \emph{lottery-based} if it allocates the remaining objects in this way. Note that allocating the remaining objects uniformly at random among agents with types in $C^{aud}\cup C^{ic}$ that did not get an object in the merit-based allocation does not yield the desired interim allocation in general because higher types in $C^{aud}$ are less likely to remain eligible for the second-stage allocation (they are more likely to receive an object under the merit-based allocation). 
However, we show that inequalities \eqref{eq:c_allo}--\eqref{eq:c_aud} ensure that such a lottery-based allocation exists:

\begin{proposition}\label{prop:opt_allo_rule}
There is a lottery-based allocation rule $p^\ell$ that, together with the merit-based allocation rule $p^m$, induces $P$. That is,  $p^m+p^\ell$ is feasible and induces $P$.
\end{proposition}

\medskip

With the allocation rule at hand, we can see that it is not implementable in dominant strategies or even as an ex-post Nash equilibrium. The following remark formalizes this.  

\begin{remark}\label{remark:EPIC}
 The optimal BIC mechanism cannot be implemented to satisfy ex-post incentive constraints.
 To see this, note that on $C^{allo}$ there is an essentially unique allocation rule that implements $P$, namely to allocate $m$ objects efficiently. Now fix an agent $i$ and a type profile $\t$ such that $m$ other agents have a type in $C^{allo}$ above $t_i$.  If agent $i$ is truthful,  he will not get an object.  If he claims to be a high type,  he will get an object whenever he is not verified.  Since not all $m$ agents can get verified, there is a profitable deviation.  
\end{remark}

\subsection{Finding the optimal allocation guarantee}\label{subsec:foc}
Theorem \ref{thm:opt} shows that the optimal mechanism lies in the class of merit-with-guarantee rules. 
The only remaining optimization problem is to establish the value of $\varphi$ that maximizes the expected payoff for the principal. In Section 3.1 we concluded that any optimal value of $\varphi$ must be at least $\frac{m-k}{n}$. Thus, by Corollary \ref{lemma:intervals} we can write  the principal's payoff from using a merit-with-guarantee rule with parameter $\varphi$ as,
\begin{align*}
U(\varphi) &:= \int_0^{\gamma_1(\varphi)} -\frac{1}{n} c_{\varphi}^{ic'}(F(t)) t \,\mathrm dF(t) + \int_{\gamma_1(\varphi)}^{\gamma_2(\varphi)} -\frac{1}{n} c^{allo'}(F(t)) t \,\mathrm dF(t) \\
& +\int_{\gamma_2(\varphi)}^{\gamma_3(\varphi)} -\frac{1}{n} c_{\varphi}^{aud'}(F(t)) t \,\mathrm dF(t)+ \int_{\gamma_3(\varphi)}^{1} -\frac{1}{n} c^{allo'}(F(t)) t \,\mathrm dF(t)
\end{align*}
The next result shows how to find an optimal value of $\varphi$.
\begin{proposition}\label{prop:opt_phi}
Every optimal choice of $\varphi$ must satisfy the first-order condition $U'(\varphi)=0$, which is given by the following condition:
\begin{align}\label{eq:FOC}
\gamma_1(\varphi)F(\gamma_1(\varphi)) + \gamma_2(\varphi)[1-F(\gamma_2(\varphi)) ]=  \gamma_3(\varphi)[1-F(\gamma_3(\varphi)) ] +  \int_0^{\gamma_1(\varphi)}  t \, \mathrm dF(t) +\int_{\gamma_2(\varphi)}^{\gamma_3(\varphi)}  t \, \mathrm dF(t)
\end{align}
\end{proposition}
One observation that follows from the first-order condition for $\varphi$  is that for the optimal mechanism, the $C^{aud}$ region is never empty: If it were empty, we would obtain $\gamma_2=\gamma_3$ and equation \eqref{eq:FOC} would simplify to 
\begin{align}
\gamma_1(\varphi)F(\gamma_1(\varphi))=    \int_0^{\gamma_1(\varphi)}  t \, \mathrm dF(t).
\end{align}
However, this condition never holds if $\gamma_1$ is strictly above the lowest type (and the density $f$ is strictly positive on the type space). 

\begin{example}
Suppose there are three agents, two objects to allocate, and only one agent can be verified: $n=3,m=2,k=1$. Types are distributed uniformly on $[0,1]$.

We will illustrate how to use Proposition \ref{prop:opt_phi} to find an optimal value of $\varphi$. Finding an optimal value of $\varphi$ corresponds to determining which constraint function is binding for the upper bound,
\begin{align*}
3 \int_{t}^1 P(x) \,\mathrm d x \le c_{\varphi}(t)=\min\{c^{allo}(t),c_{\varphi}^{aud}(t),c_{\varphi}^{ic}(t)\}.    
\end{align*}
For this example the constraint functions are  the following:
\begin{align*}
c^{allo}(t)= t^3 -3t^2 +2, \, \, c_{\varphi}^{aud}(t)= -t^3 - 3\varphi t + 1 + 3\varphi, \text{ and }c_{\varphi}^{ic}(t)=  2 - 3 \varphi t,
\end{align*}
where the optimal value of $\varphi$ lies in  $ [\frac{1}{3},\frac{2}{3}]$. It can be established that in this example $C^{allo}=[\gamma_3,1]$, giving  us the following structure of the merit-with-guarantee rule:
\begin{align*}
P(t) = \begin{cases}
\varphi &\text{, for } t\in [0,\gamma_1(\varphi))\\
t^2 +\varphi & \text{, for } t\in [\gamma_2(\varphi),\gamma_3(\varphi))\\
2t-t^2 &\text{, for } t\in [\gamma_3(\varphi),1].
\end{cases}
\end{align*}
The expected per-agent payoff for the principal as a function of $\varphi$ is 
\begin{align*}
U(\varphi) = \int_0^{\gamma_1(\varphi)} \varphi t \,\mathrm dt + \int_{\gamma_2(\varphi)}^{\gamma_3(\varphi)} t^3 +\varphi t \,\mathrm dt + \int_{\gamma_3(\varphi)}^{1} 2t^2-t^3  \,\mathrm dt.
\end{align*}
Using Proposition \ref{prop:opt_phi} we have the following condition for $\varphi$ to be optimal:
\[
\gamma_1(\varphi)=  \gamma_3(\varphi)[1-\gamma_3(\varphi) ] +  \int_0^{\gamma_3(\varphi)}  t \, \mathrm d t.
\]
Solving this equation gives us the optimal value of $\varphi^*=0.34764$ and an expected payoff of $1.223$. This can be compared with the expected payoff of 1 from using a random allocation rule where all agents gets an object with an equal probability of $\frac{1}{3}$. Another comparison is with the first best mechanism that always allocates the two objects to the agents with the highest types yielding an expected payoff of $1.25$.

\end{example}

 \section{Why a merit-with-guarantee rule is optimal}\label{sec:proof_red_optimal}
We next discuss the feasibility constraints---implying that the principal can allocate at most $m$ objects and learn the types of at most $k$ agents---in terms of interim rules.
Recall that we consider only symmetric interim rules throughout and that for $P$ to be \emph{feasible} there must exist an allocation rule $p$ with $P(t)=\Eis[p_i(t,\t_{-i},s)]$ for each $t\in T$.  To formally state the constraints for $P$ to be feasible, we need some notation. Given $E\subseteq T$, let $I(\t,E) := \{i\in N| t_i \in E \}$ denote the set of agents with types in $E$ that are compatible with $\t$. 

\begin{lemma}\label{lemma:border_alloc}
 An interim allocation rule $P$ is feasible if and only if for all Borel sets $E\subseteq T$, 
\begin{equation}\label{eq:border_allocation}
  n\int_E P(t_i) dF(t_i) \le  \int \min\{I(\t,E),m\} d\F(\t).
 \end{equation} 
Moreover, if $P$ is non-decreasing, it is enough to check this inequality for sets $E$ of the form $E=[e,1]$.
\end{lemma}

Lemma~\ref{lemma:border_alloc}, which follows immediately from Corollary 4(ii) in \citeasnoun{che13}, characterizes when an interim allocation rule is feasible. 

Note that for an interim audit rule $A$ to be feasible there must exist an audit rule $a$ with $A(t)=\Eis[a_i(t,\t_{-i},s)]$ for each $t\in T$,  and $a(\t,s)\le p(\t,s)$, for the specific  $p$ that induces $P$. We imposed the latter constraint $a(\t,s)\le p(\t,s)$, since this simplified our characterization of incentive compatibility. However, this additional constraint complicates the characterization of which interim audit rules are feasible given an allocation rule. To state the characterization, let $p$ be a symmetric and deterministic allocation rule.\footnote{With slight abuse of notation, we will drop the realization of the randomization device as an argument whenever it plays no role.} 
Define for every $E\subseteq T$ and every $\t\in \T$, $J^p(\t,E) := \{i\in N| t_i \in E \text{ and } p_i(\t)=1\}$. 

\begin{lemma}\label{lemma:border_audit}
An interim audit rule $A$ is feasible given a symmetric and deterministic ex-post allocation rule $p$ if and only if for all Borel sets $E  \subseteq T$, 
\begin{equation}\label{eq:border_audit}
  n \int_{E} A(t_i) dF(t_i) \le  \int \min\{|J^p(\t,E)|,k\} d\F(\t).
 \end{equation}
\end{lemma}

Observe that the condition $a(\t,s)\le p(\t,s)$ imposes constraints on which agents can be audited. \citeasnoun{che13} develop a powerful approach to characterize feasible interim rules that allow for complicated constraints on which subsets of agents can obtain an object based on their identities (for example, allowing for a constraint that at least one object is allocated to agents in $\{1,2\}$ etc.). We cannot use their results directly because the set of agents that can be audited depends not only on the identity of the agents but also on the realized type profile.  
However, adapting the proof technique used in \citeasnoun{che13} we obtain the above characterization of feasible interim audit rules. 

\medskip
We now simplify the principal's problem \eqref{P_int}. First, we formulate a parametrized linear problem by replacing the incentive constraint by $A(t)\ge P(t)-\varphi$, where $\varphi$ is a parameter. Second, we can assume that this constraint is always binding, i.e., for all $t$ it holds that $A(t)= P(t)-\varphi$, because we can always lower the audit probability without violating any other constraint or changing the objective value. Finally, we have already argued that \eqref{eq:upper_bound} is a necessary condition for $P$ to be feasible for problem \eqref{P_int}, and we only impose this necessary condition to obtain the following problem:
\begin{align}\label{relaxed_problem}
\max_{P}\  n \, \E_{t}[P(t) t]&\tag{R}\\
\text{s.\ t. } n \int^1_t P(x) dF(x)&\le c_{\varphi}(F(t)) \text{ for all } t\in T \label{eq:R_constraint} \\
\varphi&\le P(t) \text{ for all } t\in T \label{eq:IC}
\end{align}

We show that for any parameter $\varphi$, this relaxed problem \eqref{relaxed_problem} is solved by a merit-with-guarantee rule. 

\begin{lemma}\label{lemma:opt_P}
There is a merit-with-guarantee allocation rule that solves problem \eqref{relaxed_problem}.
\end{lemma}
The optimization problem \eqref{relaxed_problem} is indeed a relaxation of the principal's original problem \eqref{P_int}. Because take any solution to the original problem $(P^*,A^*)$ and let  $\varphi:=\inf_{t'}P^*(t')$. Then $P^*$ is feasible for the problem \eqref{relaxed_problem} for the parametrization $\varphi$ and the objective value coincides with the solution $(P^*,A^*)$ in the original problem \eqref{P_int}. This gives a lower bound on the objective value in the relaxed problem \eqref{relaxed_problem}. Therefore, if a merit-with-guarantee rule, the solution to the relaxed problem,  also is feasible in the original problem,  it has to be the optimal mechanism for the principal.     

\medskip

What remains to prove is that the solution to problem \eqref{relaxed_problem} together with its associated audit rule is a feasible solution for the original problem \eqref{P_int}. Thereby, implying that a merit-with-guarantee rule solves the original problem and is therefore the optimal mechanism for the principal. The difficulty here lies in showing that there exists a pair $(p,a)$ of feasible allocation and audit rules inducing the merit-with-guarantee rule. To do this we begin with establishing a Border-like result, Lemma \ref{lemma:border_general}, on the feasibility of interim rules with type-dependent ex-post constraints. This result is required both to prove the existence of the lottery rule $p^\ell$ given the merit rule $p^m$, and the audit rule $a$. For the lottery rule $p^\ell$ the number of objects available at the lottery stage varies across different profiles $\t$ and an agent with a type in $C^{aud}$ is only eligible to enter the lottery if he has not yet been allocated an object. Similarly, to prove the feasibility of the interim audit rule we require type-dependent ex-post constraints, Lemma \ref{lemma:border_audit} provides the precise result we need. 

Although the merit-with-guarantee interim rules are monotone, it does not suffice to check  feasibility constraints in Lemma \ref{lemma:border_audit} and \ref{lemma:border_general} only for upper sets. This is because the number of objects available to allocate (or the number of audits available) is effectively dependent on the type profile: how many objects are allocated in the merit stage depends on the type profile, and therefore the number of objects available in the lottery stage; similarly, because our formulation of incentive compatibility assumes that only agents receiving an object can be audited, the number of audits that can be conducted depends on the type profile. If the number of objects is type-dependent, it is not sufficient to check the border constraints only for upper sets.\footnote{Indeed, consider an example with two agents and two types each, $T=\{\underline{t}, \overline{t}\}$ with $ \underline{t}<\overline{t} $. Types are iid and both types are equally likely.  Further,  all agents can be assigned at most one object at all type profiles but the number of objects available at a given type profile varies.  The following table specifies how many objects are available:
\medskip
\begin{tabular}{c||c|c}
 &$ \underline{t}_2$ &$ \overline{t}_2$ \\
\hline
\hline
$ \underline{t}_1$ &0&1\\
\hline
$ \overline{t}_1$&2&0\\
\hline
\end{tabular}
\medskip

Now, consider the following monotone interim rule $P_1(\underline{t}_1)=P_1(\overline{t}_1)= P_2(\underline{t}_2)=0.25 $ and  $ P_2(\overline{t}_2)=0.5 $.  One can verify that the Border constraints are satisfied for all upper sets, but violated for the set $\{\underline{t},_1 \overline{t}_2 \}$. 
Thus, the interim rule is not feasible. }

\section{Conclusion}\label{sec:con}

We have analyzed the problem of allocating $m$ identical objects when only $k$ checks are available to verify agents' private information and monetary transfers cannot be used. The merit-with-guarantee rule maximizes the expected value from allocating the objects. This mechanism combines an efficient allocation rule that allocates to agents with high types, the merit-based part, with a second-stage lottery. A merit-with-guarantee rule is parametrized by a guarantee  $\varphi$ which determines the expected probability of winning an object at the lottery stage.  We provide a first-order condition that helps to find the optimal guarantee $\varphi$.  
 
 Our model takes as given that there is a fixed number $k$ of checks available to the principal.  Suppose instead that we could add checks at a cost before constructing the optimal mechanism.  An interesting question to ask is what is the optimal number of checks given a cost function for adding additional checks? 

There are other environments where capacity-constrained verification may arise. For instance, consider a situation where agents arrive sequentially over time and at each time an allocation decision together with a verification decision has to be made.  Suppose further that there is a finite time horizon,  a given number of objects to allocate,  and a fixed number of checks that can be carried out over the whole time period.  This is one example of another model where capacity-constrained verification may arise naturally.  Questions about the structure of optimal mechanisms in such a dynamic setting would be interesting to study.  Another example is collective decisions with capacity-constrained verification. These are also settings where a fixed number of checks can be the relevant constraint on the verification technology.

\newpage
\part*{}
\begin{appendix}
\section{Appendix}
\subsection{Characterization of BIC mechanisms}\label{app_BIC}

\begin{proof}[Proof of Lemma \ref{lemma:bic_simple}]
Let $i\in N$. By definition of BIC a mechanism $(p,a)$ in equation~\eqref{eq:BIC_def} and the fact that $u_i(t_i)>0$ for all $t_i\in T$  we get that agent $i$ with type $t_i\in T$ has no incentive to deviate if and only if, for all $t'_i\in T$,
\begin{align}\label{eq:bic_simp}
     \E_{\t_{-i},s}[p_i(t_i,\t_{-i},s)] &\ge \E_{\t_{-i},s}[p_i(t'_i,\t_{-i},s) [1-a_i(t'_i,\t_{-i},s)]].
\end{align}
Since $a_i(\t,s)=1$ only if $p_i(\t,s)=1$, the right-hand side equals $P(t_i')-A(t_i')$.
Since~(\ref{eq:bic_simp}) is required to hold for all $t_i\in T$, it must in particular hold for the infimum over $T$. Thus,  Definition~\ref{def:BIC} of BIC is equivalent to the equation in Lemma~\ref{lemma:bic_simple}. 
\end{proof}


\subsection{Characterization of feasible interim rules with type-dependent ex-post constraints}\label{app_reduc_ver_feasible}

We need to introduce some notation and definitions that we will use throughout the appendix. 

In contrast to the standard union the disjoint union indexes each element by the set it comes from. That is, for $D_1,D_2\subseteq [0,1]$, the disjoint union is defined as $D_1\sqcup D_2:= \{ (t,1)| t\in D_1\}\cup \{(t,2)| t\in D_2\}$, and analogously for disjoint unions of $n$ sets.

The next lemma characterizes interim rules in our setting with type-dependent ex-post constraints. Specifically, we assume that if type profile $\t$ realizes, $h(\t)$ objects can be allocated. Moreover, we assume only agents in the set $J(\t)$ can get an object. An allocation rule $p=(p_i)_{i\in N}:\T \rightarrow [0,1]^n$ is \emph{feasible} if, for all $\t\in\T$, $\sum_i p_i(\t)\le h(\t)$ and $p_i(\t)>0$ implies $i\in J(\t)$.
To formulate the lemma we need to consider  asymmetric settings  with cdfs $F_i:T_i\rightarrow[0,1]$ and defining correspondingly asymmetric interim rules $P_i:T_i\rightarrow [0,1]$ where   
$P_i(t):= \E_{\t_{-i}}[p_i(t,\t_{-i})]$.

\begin{lemma}\label{lemma:border_general}
Let $h(\t)$ be the total number of objects available, and let $J(\t)$ denote the set of agents that can receive an object,  which can both depend on the realized type profile. There is a feasible allocation rule $p=(p_i)_{i\in N}$ inducing an interim allocation rule $P=(P_i)_{i\in N}$ if and only if, for all  sets  $E=\bigsqcup\limits_{i\in N} E_i$, where each $E_i$ is a Borel subset of $T_i$, 
\[ \sum_i \int_{E_i} P_i(t) dF_i(t) \le \int \min\{|J(\t)\cap I(\t,E)|,h(\t)\} d\F(\t). \]     
\end{lemma}

We show that the same arguments as in \citeasnoun{che13} establish the claim even with type-dependent constraints as in Lemma \ref{lemma:border_general}. Alternatively, the result follows from a recent result in \citeasnoun[Theorem 5 for finitely many types]{valenzuela2022greedy} combined with an approximation argument.

\begin{proof}[Proof of Lemma \ref{lemma:border_general}]
We first prove the result for finite type spaces and then extend the result to continuous type spaces via an approximation argument. To this end, suppose first that $T$ is finite.

Construct a network as in  \citeasnoun[p. 2495f]{che13}. The upper capacity of a supply node is given by $k(\t,N'):= \f(\t) \min\{|I(\t, N'\cap D)\cap J(\t)|,h(\t)\}$. All lower capacities are set to $0$, and all remaining capacities are defined as in \citeasnoun{che13}.
    
Next, verify that $k$ is submodular. To do so, it is enough to show that for any $N'\subset N'' \subset  N $, and $n\in N\setminus N''$,
\[
k(\t,N'\cup\{n\})-k(\t,N')\geq k(\t,N''\cup\{n\})-k(\t,N''),
\]
which is equivalent to 
\[
\begin{split}
& \min\{|I(\t, (N'\cup\{n\}) \cap D)\cap J(\t)|,h(\t)\}- \min\{|I(\t, N' \cap D)\cap J(\t)|,h(\t)\}\geq\\
& \min\{|I(\t, (N''\cup\{n\}) \cap D)\cap J(\t)|,h(\t)\}- \min\{|I(\t, N'' \cap D)\cap J(\t)|,h(\t)\}.  
\end{split}
\]
Note that the right-hand side of this inequality takes on the value 0 or 1. If it takes on the value 1, then the left-hand side must also take on the value 1, and the inequality follows. 

 Therefore, analogous to Theorem 1 in \citeasnoun{che13}, it follows that an interim rule $P$ is feasible if and only if there is a feasible flow. Since $k$ is submodular, the constraints are paramodular, and the same argument as in their Theorem 3 completes the proof.

 Finally, the approximation argument used to prove Theorem 5 in \citeasnoun{che13} establishes the conclusion for our setting with infinitely many types.
\end{proof}

\subsection{A merit-with-guarantee allocation rule solves the relaxed problem}

\begin{proof}[Proof of Lemma~\ref{lemma:opt_P} ]
Note that $c^{allo}(1)=c^{aud}_{\varphi}(1)=0$, $c^{allo}(0)=c^{ic}_{\varphi}(0)=m$, and all $c^{i}$'s are non-decreasing and concave. Hence, $c$ is concave and $P$ is well-defined and non-decreasing.

\emph{Step 1: Feasibility} 

\noindent
Given arbitrary $t\in T$, 
\begin{align}
 n \int_{t}^1 P(x) dF(x)= n\int_t^1 (-\frac{1}{n}c'_\varphi(F(x)))dF(x) = c_\varphi(F(t))-c_\varphi(F(1)). \label{eq:P_equality}
\end{align}
 Since $c_\varphi(1)=0$, we conclude that \eqref{eq:R_constraint} holds.

It follows from the definition of $c_{\varphi}$ that $c_{\varphi}'(0)\le -n\varphi$ (since $c^{allo}(0)=m=c_{\varphi}^{ic}(0)$ and the derivatives of $c_{\varphi}^{ic}$ and $c_{\varphi}^{aud}$ are less than $-n\varphi$). Because $c_{\varphi}$ is concave, we obtain
\[P(t)\equiv -\frac{1}{n} c'(F(t))\ge \varphi \]
for all $t\in T$.
Thus, \eqref{eq:IC} holds, and we conclude that $P$ is feasible for problem \eqref{relaxed_problem}. 

\medskip

\emph{Step 2: Optimality}

\noindent
We first argue that any feasible interim rule $\tilde{P}$ satisfies $n\int_{t}^1 \tilde{P}(x) dF(x) \le m- n \varphi F(t)$.
Indeed, since $\tilde{P}(t)\ge \varphi$ for all $t$, \eqref{eq:border_allocation} implies $n \int_0^{t} \varphi dF(x) +   n\int_{t}^1 \tilde{P}(x) dF(x) \le m$, or 
\begin{equation} 
n\int_{t}^1 \tilde{P}(x) dF(x) \le m- n \varphi F(t). \label{eq:border_plus_bic}
\end{equation}

Given any feasible interim rule $\tilde{P}$, we use integration by parts to get the following upper bound for the objective function in \eqref{relaxed_problem}:
  \begin{align*}
    &n \int_{0}^{1} f(t) \ \tilde{P}(t)\ t\ dt \\
    = & n \left. t  \int_{0}^{t} f(x)\ \tilde{P}(x)\  dx \right|_{t=0}^{1} - n \int_{0}^{1} \int_{0}^{t} f(x)\ \tilde{P}(x)\  dx \ dt \\
    = & n \int_0^{1}  \ \int_{0}^{1} f(x)\ \tilde{P}(x)\  dx \ dt - n \int_{0}^{1} \int_{0}^{t} f(x)\ \tilde{P}(x)\  dx \ dt \\
    = &n \int_{0}^{1} \int_{t}^{1} f(x)\ \tilde{P}(x)\  dx \ dt \\
    \le & \int_0^1 c_\varphi(F(t)) d t
  \end{align*}
where the last inequality follows from \eqref{eq:R_constraint}. It follows from \eqref{eq:P_equality} that the last inequality holds as an equality for $P$.  Therefore, $P$ is an optimal solution to \eqref{relaxed_problem}.
\end{proof}

\subsection{Properties of the $c$ function and the constraint set $C^i$'s}\label{App:prop_c}

\renewcommand{\r}{\mathbf{r}}

\begin{lemma}\label{lemma:c_functions}
Let $q=F(t_i)$. The derivatives of $c^{allo}$ and $c_\varphi^{aud}$ are given by,
\begin{align}
\frac{d}{dq}(c^{allo}(q)) &= -n \sum_{j=0}^{m-1} \binom{n-1}{j} q^{n-1-j} (1-q)^{j}\label{eq:c_1_diff} \\
\frac{d}{dq}(c^{aud}_\varphi(q)) &= -n \sum_{j=0}^{k-1} \binom{n-1}{j} q^{n-1-j} (1-q)^{j} -n\varphi  \label{eq:c_2_diff}.
\end{align}
Also, $c^{allo}(0)=c^{ic}_{\varphi}(0)=m$ and $c^{allo}(1)=c^{aud}_{\varphi}(1)=0$. If $\varphi\ge \frac{m-k}{n}$ then $c^{aud}_\varphi(0)\ge c^{allo}(0)$.
\end{lemma}
\begin{proof}[Proof of Lemma \ref{lemma:c_functions}]
Note first that,  
\begin{align}
 \int \min\{(|I(\t,[t_i,1])|,m)\  d\F(\t)&=\sum_{i=1}^n\min\{i,m\}\binom{n}{i}(1-q)^{i}q^{n-i} = c^{allo}(q) \\
 \int \min\{|(I(\t,[t_i,1])|,k)\  d\F(\t)&=\sum_{i=1}^n\min\{i,k\}\binom{n}{i}(1-q)^{i}q^{n-i} = c^{aud}_\varphi(q)- n(1-q)\varphi
\end{align}
where $q=F(t_i)$. We will prove equality \eqref{eq:c_1_diff} for  $c^{allo}$, and an analogous argument can be used to prove equality \eqref{eq:c_2_diff} for $c_\varphi^{aud}$.

Let $G(i;n,\rho)$ denote the CDF of the binomial distribution, where $i$ is the number of success in $n$ independent Bernoulli trials, and $\rho$ is the probability of success.  The CDF of the  binomial distribution can be written as,
\begin{align*}
G(i;n,\rho) &= (n-i) \binom{n}{i}\int_0^{1-\rho} t^{n-i-1} (1 - t)^{i} \ dt\\
\text{ and }\frac{\partial G(i;n,\rho)}{\partial \rho} &= -(n-i)\binom{n}{i}(1-\rho)^{n-i-1}\rho^i= -n \binom{n-1}{i}(1-\rho)^{n-i-1}\rho^i.
\end{align*}
Now, using that the probability of success $\rho=1-F(t_i)=1-q$ 
we can express $c^{allo}(q)$ using the binomial distribution,
\begin{align*}
c^{allo}(q)=\sum_{i=1}^m i [G(i;n,1- q)-G(i-1;n,1- q)]+ m [1-G(m;n,1- q)].
\end{align*}
It then follows that, 
\begin{align}
G(i;n,1-q)&=(n-i) \binom{n}{i}\int_0^q t^{n-i-1} (1 - t)^{i} \ dt\nonumber \\
\label{eq:c1_Gdiff} \frac{\partial G(i;n,1-q)}{\partial q} &= n\binom{n-1}{i}q^{n-i-1}(1-q)^i.
\end{align}
To establish equation \eqref{eq:c_1_diff} that gives us the derivative of $c^{allo}(q)$, we use induction over the number of objects $m$.
The claim in \eqref{eq:c_1_diff} is trivially true for $m=1$.
\medskip

\textit{Induction hypothesis:}  $\frac{d}{dq}(c^{allo}(q))= -n\sum_{i=0}^{m-1} \binom{n-1}{i} q^{n-1-i} (1- q)^{i}  $ for all $m$.
\smallskip

\noindent
We are going to show that if the induction hypothesis holds for $m$ then it also holds for $m+1$.
We will show this by looking at the difference between $c^{allo}(q|m+1 \text{ objects})$ and $c^{allo}(q|m \text{ objects})$ and then differentiating wrt to $q$ using equation \eqref{eq:c1_Gdiff}. That is,  
\begin{align*}
c^{allo}(q|m+1 )- c^{allo}(q|m )=& \sum_{i=1}^{m+1} i [G(i;n,1- q)-G(i-1;n,1- q)]+ (m+1 )[1-G(m+1;n,1- q)]\\
- &\Big[ \sum_{i=1}^m i [G(i;n,1- q)-G(i-1;n,1- q)]+ m [1-G(m;n,1- q)]\Big] \\
=& 1-G(m;n,1-q)
\end{align*}
Now, differentiating with respect to $q$ gives,
\begin{align*}
\frac{d}{dq}\big(c^{allo}(q|m+1 )- c^{allo}(q|m )\big)=&-\frac{\partial G(m;n,1-q)}{\partial q} =- (n-m)\binom{n}{m}q^{n-m-1}(1-q)^{m}\\
=&-n \binom{n-1}{m}q^{n-m-1}(1-q)^{m}.
\end{align*}
By the induction hypothesis, $\frac{d}{dq}(c^{allo}(q))= -n\sum_{i=0}^{m-1} \binom{n-1}{i} q^{n-1-i} (1- q)^{i} $. Thus,  from the above equality, we get that also for $m+1$ objects equation \eqref{eq:c_1_diff} holds.
In an analogous argument, also using proof by induction, we can show that equation \eqref{eq:c_2_diff} for $\frac{d}{dq}(c_\varphi^{aud}(q))$ holds. 

Note also that,
$ c^{allo}(0)=m \text{ and } c^{aud}_\varphi(0)=k+n\varphi $. Thus, $c^{aud}_\varphi(0)\geq c^{allo}(0)$ for $\varphi\geq \frac{m-k}{n}$ as claimed in the lemma.
\end{proof}

 \begin{lemma}\label{lemma:intervals_new}
 Let $\varphi\in[0,\frac{m}{n}]$.  
 This gives us the following four different structures of the lower envelope of $c$.
\begin{itemize}
\item[(1)] For $\varphi\in[0,\frac{m-k}{n}]$ there is $r(\varphi)\in[0,1]$ such that: $c_{\varphi}^{aud}(r(\varphi))=c^{allo}(r(\varphi))$, 

  $C^{ic}=\emptyset$,  $C^{aud}=[0,r(\varphi)]$,  and  $C^{allo}=[r(\varphi),1]$.
  
\item[(2)] For $\varphi\in (\frac{m-k}{n},\overline{\varphi}]$ there are $r_1(\varphi), r_2(\varphi)$ with   $0<r_1(\varphi)<r_2(\varphi)<1$,  and $z_1(\varphi),z_2(\varphi)\in [0,1]$ such that: 

$c_{\varphi}^{aud}(r_i(\varphi))=c^{allo}(r_i(\varphi))$,  for $i=1,2$,  $c_{\varphi}^{ic}(z_1(\varphi))=c^{allo}(z_1(\varphi))$,  $c_{\varphi}^{ic}(z_2(\varphi))=c_{\varphi}^{aud}(z_2(\varphi))$.

\begin{itemize}

\item [(2.a)] If $z_1(\varphi)< r_1(\varphi) $, then

   $C^{ic}=[0,z_2(\varphi)]$,  $C^{allo}=[z_2(\varphi), r_1(\varphi)]\cup [r_2(\varphi),1]$, and $C^{aud}=[r_1(\varphi),r_2(\varphi)]$.   

\item [(2.b)] If If $z_1(\varphi)\geq r_1(\varphi) $, then 

 $C^{ic}=[0,z_2(\varphi)]$,  $C^{aud}=[z_2(\varphi),r_2(\varphi)]$,  and  $C^{allo}=[r_2(\varphi),1]$.

\end{itemize}
\item [(3)]  For $\varphi\in(\overline{\varphi},\frac{m}{n}]$   there is $z_1(\varphi)\in(0,1]$ such that:
$c_{\varphi}^{ic}(z_1(\varphi))=c^{allo}(z_1(\varphi))$,   
 $C^{ic}=[0,z_1(\varphi)]$,  $C^{allo}=[z_1(\varphi), 1]$, and $C^{aud}=\{1\}$.  
\end{itemize}
 \end{lemma}
\begin{proof}
Let us first look at the second derivatives of $c^{allo}$ and $c_{\varphi}^{aud}$, and let  $q=F(t)$.  It follows from Lemma \ref{lemma:c_functions} that 
\begin{align*}
\frac{d}{dq}(c^{allo}(q))= -n \sum_{j=0}^{m-1} \binom{n-1}{j} q^{n-1-j} (1-q)^{j}=-n G(m-1,n-1,1-q).
\end{align*}
Thus, we obtain that,
\begin{align*}
\frac{d^2}{dq^2}(c^{allo}(q)) &= -n \frac{\partial G(m-1,n-1,1-q)}{\partial q} = -n(n-m)\binom{n-1}{m-1} q^{n-m-1} (1-q)^{m-1} \\ 
&= -nm\binom{n-1}{m} q^{n-m-1} (1-q)^{m-1}.
\end{align*}
Again,  using Lemma \ref{lemma:c_functions} for $c_\varphi^{aud}$ we have:
\begin{align*}
\frac{d}{dq}(c_{\varphi}^{aud}(q))= -n \sum_{j=0}^{k-1} \binom{n-1}{j} q^{n-1-j} (1-q)^{j}-n\varphi=-n G(k-1,n-1,1-q) -n\varphi.
\end{align*}
Analogously to above we obtain that,
\begin{align*}
\frac{d^2}{dq^2} (c^{aud}_\varphi(q) &= -nk\binom{n-1}{k} q^{n-k-1} (1-q)^{k-1}.
\end{align*}
 Both second derivatives are negative and we conclude that $c^{allo}$ and $c_{\varphi}^{aud}$ are concave functions.  Since $c_{\varphi}^{ic}$ is a linear function,   $c_{\varphi}^{ic}-c^{allo}$ and $c_{\varphi}^{ic}-c_{\varphi}^{aud}$ crosses zero at most once. 

Since  $c_{\varphi}^{ic}(0)=m=c^{allo}(0)$ and $c_{\varphi}^{ic}(1)\geq 0=c^{allo}(1)$, there is  $z_1(\varphi)\in[0,1] $ such that:
\begin{align}\label{eq:ic_allo}
c_\varphi^{ic}(F(t))  &\le c^{allo}(F(t))    &\text{ for all } t\in [0, z_1(\varphi)] \nonumber\\
 c^{allo}(F(t)) &\le c_{\varphi}^{ic}(F(t))   &\text{ for all } t\in [z_1(\varphi), 1].  
\end{align}
Where, $z_1(0)=0$ and $z_1(\frac{m}{n})=1$ and for all other $\varphi\in(0,\frac{m}{n})$ we have that $z_1(\varphi)\in(0,1)$.

For $c_\varphi^{ic}- c_\varphi^{aud}$,  there is $z_2(\varphi)\in[0,1]$ such that: 
\begin{align}\label{eq:ic_aud}
c_\varphi^{ic}(F(t))  &\le c_\varphi^{aud}(F(t))    &\text{ for all } t\in [0, z_2(\varphi)]\nonumber  \\
 c_\varphi^{aud}(F(t)) &\le c_\varphi^{ic}(F(t))   &\text{ for all } t\in [z_2(\varphi), 1]. 
 \end{align}
Where,  $z_2(\varphi)=0$ if $\varphi\in[0,\frac{m-k}{n}]$ and $z_2(\varphi)\in (0,1]$ if $\varphi\in(\frac{m-k}{n},\frac{m}{n}]$. 
Because if $\varphi\in[0,\frac{m-k}{n}]$, then $c_{\varphi}^{aud}(0)\leq m=c_{\varphi}^{ic}(0)$, and since we always have that $\frac{d}{dq}(c_{\varphi}^{aud}(q))<\frac{d}{dq}(c_{\varphi}^{ic}(q))$ for all $q<1$.  For the other case with  $\varphi\in(\frac{m-k}{n},\frac{m}{n}]$ we have that  $c_{\varphi}^{aud}(0)>m=c_{\varphi}^{ic}(0) $ and  $c_{\varphi}^{aud}(1)=0\leq c_{\varphi}^{ic}(1)$ implying that $z_2(\varphi)\in(0,1]$.

Next,  let us look at  $c^{allo}-c_{\varphi}^{aud} $.  The difference between two concave functions can have arbitrary many crossings.  However, by examining the number of zeros of the second derivative of  $c^{allo}-c_{\varphi}^{aud} $ we will show that there can be at most three zeros of $c^{allo}-c_{\varphi}^{aud} $ on $[0,1]$.  
Now,  consider the second derivative of  $c^{allo}-c_{\varphi}^{aud} $,
\begin{align*}
\frac{d^2}{dq^2}  (c^{allo}(q)-c^{aud}_\varphi(q)) &= -n q^{n-m-1}(1-q)^{k-1} \Big[ m\binom{n-1}{m} (1-q)^{m-k} - k \binom{n-1}{k} q^{m-k} \Big].
\end{align*}
Note that the term in brackets is decreasing in $q$ and, hence,  $\frac{d^2}{dq^2}  (c^{allo}(q)-c^{aud}_\varphi(q))$ crosses zero at most once. It follows that $ \frac{d}{dq}(c^{allo}(q)-c^{aud}_\varphi(q))$ has at most two zeros on $[0,1]$,  which implies that $c^{allo}-c^{aud}_\varphi$ has at most three zeros on $[0,1]$.  Since $c^{allo}(1)=c^{aud}_\varphi(1)$, there are at most two zeros of $c^{allo}-c^{aud}_\varphi$ on $[0,1)$.
Note further that $(i)$: for all $\varphi\in(0,\frac{m}{n}]$,  $\frac{d}{dq}(c^{allo}(q)) > \frac{d}{dq} (c^{aud}_\varphi(q))$ for $q=1$,  and $(ii)$: 
$c^{allo}(q)\le c^{aud}_\varphi(q)$ for all $q$ close to 1.  Thus,  $c^{allo}(q)<c_{\varphi}^{aud}(q)$ for $q$ close to 1, for any $\varphi\in (0,\frac{m}{n}]$.

Summarizing,  there are numbers $r_1(\varphi),r_2(\varphi)\in[0,1]$ with $r_1(\varphi)\le r_2(\varphi)$ such that:
\begin{align}\label{eq:allo_aud}
c^{allo}(F(t))  &\le c_\varphi^{aud}(F(t))    &\text{ for all } t\in [0, r_1(\varphi)]\nonumber  \\
 c_\varphi^{aud}(F(t)) &\le c^{allo}(F(t))   &\text{ for all } t\in [r_1(\varphi), r_2(\varphi)]\\
 c^{allo}(F(t)) &\le c_\varphi^{aud}(F(t))   &\text{ for all } t\in [r_2(\varphi), 1]. \nonumber
 \end{align}
Where,  $r_1(\varphi)=0$ if $\varphi\in(0,\frac{m-k}{n}]$, since $c_{\varphi}^{aud}(0)\le c^{allo}(0) $ for $\varphi\in(0,\frac{m-k}{n}]$,  $c^{allo}(q)<c_{\varphi}^{aud}(q)$ for $q$ close to 1, and there cannot be more than two zeros of $c^{allo}-c_{\varphi}^{aud}$ in $(0,1)$.  If $\varphi=0$, then $c_{\varphi}^{aud}(q)\leq c^{allo}(q)$ for all $q$ and $r_1(\varphi)=0$ and $r_2(\varphi)=1$.  
Further,  $0<r_2(\varphi)<1$ for all $\varphi\in(0,\frac{m}{n}]$  since $c^{allo}(q)< c_{\varphi}^{aud}(q)$ for small $q$ and $c^{allo}(1)= c_{\varphi}^{aud}(1)$ for  all  $\varphi$.
Now, we can fully describe the lower envelope of $c(q)=\min\{c_{\varphi}^{ic}(q), c_{\varphi}^{aud}(q), c^{allo}(q)  \}$ for all $\varphi\in [0,\frac{m}{n}]$. 
 
 First, suppose $\varphi\in[0,\frac{m-k}{n}]$, then by \eqref{eq:ic_aud} we know that $c_{\varphi}^{aud}(q)\le c_{\varphi}^{ic}(q)$ for all $q$ and by \eqref{eq:allo_aud} that $r_1(\varphi)=0$.  Thus,  there is $C^{ic}=\emptyset$,  $C^{aud}=[0,r_2(\varphi)]$,  and $C^{allo}=[r_2(\varphi),1]$, as desired. 
 
 Let us  instead look at when   $\varphi\in(\frac{m-k}{n},\frac{m}{n}]$. Recall first that $c_{\varphi}^{aud}(0)>m=c^{allo}(0)$ and $c_{\varphi}^{aud}(q)>c^{allo}(q)$ for $q$ close to 1.  
 Suppose that is at least one $q\in (0,1)$ with $c_{\varphi}^{aud}(q)< c^{allo}(q)$.  Then there must be exactly two zeros of $c^{allo}-c_{\varphi}^{aud}$ in $(0,1)$.  That is,  there is $r_1(\varphi)\in( 0,1 )$ and $r_2(\varphi)\in (0,1)$ with $r_1(\varphi)<r_2(\varphi)$,  and $c_{\varphi}^{aud}(r_i(\varphi))=c^{allo}(r_i(\varphi))$,  for $i=1,2$.   Further,  if  $r_1(\varphi) <z_1(\varphi)$ then  $r_1(\varphi) \leq z_2(\varphi)  <z_1(\varphi)$,  and if  $r_1(\varphi) \leq z_1(\varphi)$, then $ r_1(\varphi) \leq z_2(\varphi) \leq z_1(\varphi)  $.  This corresponds to case $(2.a)$ and $(2.b)$, respectively. 
 
 Since $c_{\varphi}^{aud}$ is increasing in $\varphi$ it is possible that for some $\overline{\varphi}$,  $c_{\varphi}^{aud}(q)\geq c^{allo}(q)$ for all $q<1$.   That is,  there is no interior zero of  $c^{allo}-c_{\varphi}^{aud}$ for all $\varphi\geq \overline{\varphi}$, as desired. 
\end{proof}

\subsection{The sum of two allocation rules $p^m+ p^\ell$ induces $P$}\label{app:opt_alloc_rule}
We say that the disjoint union
$D_1\sqcup\dots \sqcup D_n$ is \textit{symmetric} if $D_i=D_j$ for all $i,j$. 
Let $H$ be a set. A function $f:2^H\rightarrow \R$ is \textit{symmetric} if for all $D_1\sqcup\dots \sqcup D_n$ and $i,j$,
\[f(D_1\sqcup \dots \sqcup D_i\sqcup\dots \sqcup D_j\sqcup \dots\sqcup D_n)=f(D_1\sqcup \dots \sqcup D_j\sqcup\dots \sqcup D_i\sqcup \dots \sqcup D_n).\]

\begin{lemma}\label{lemma:symmetric_solution}
    Let $H:=\bigsqcup_i [0,1]$ and $\mathcal{B}(H)$ denote the subsets $H_1\sqcup...\sqcup H_n$ of $H$ such that $H_i$ is a Borel subset of $[0,1]$ for each $i$. Let $f:\mathcal{B}(H)\rightarrow \R$ be supermodular\footnote{A function $f:\mathcal{B}(H)\rightarrow \R$ is supermodular if $X,Y,Z\in \mathcal{B}(H)$, $X\subseteq Y$, and $Z\subseteq H\setminus Y$ imply \[f(X\cup Z)-f(X)\le f(Y\cup Z)-f(Y).\]} and symmetric.
    Consider $G\subseteq \mathcal{B}(H)$ that is closed under permutations, unions, and intersections. 
    
    If $D_1\sqcup ...\sqcup D_n \in G$ maximizes $f$ on $\mathcal{B}(H)$ then there is $D'\sqcup D'\sqcup ...\sqcup D'\in G$ that also maximizes $f$. In particular, 
 there is a symmetric maximizer.
\end{lemma}

\begin{proof}
Let $D_1\sqcup \dots \sqcup D_n \in G$ be a maximizer of $f$.
Either 
\[ f((D_1\cap D_2)\sqcup (D_1\cap D_2) \sqcup D_3\sqcup\dots \sqcup D_n) \ge f(D_1\sqcup D_2\sqcup D_3\sqcup\dots \sqcup D_n) \]
or 
\begin{align*}
&f((D_1\cap D_2)\sqcup (D_1\cap D_2) \sqcup D_3\sqcup\dots \sqcup D_n) < f(D_1\sqcup D_2\sqcup D_3\sqcup\dots \sqcup D_n) \\
= &f(D_2 \sqcup D_1\sqcup D_3\sqcup\dots \sqcup D_n) 
\end{align*}
where the equality uses symmetry of $f$. 

In the latter case, by denoting $X= (D_1\cap D_2)\sqcup (D_1\cap D_2) \sqcup D_3\sqcup\dots  \sqcup D_n$, $Y= D_1\sqcup D_2\sqcup D_3\sqcup \dots \sqcup D_n$ and $Z=(D_2\setminus D_1)\sqcup (D_1\setminus D_2)\sqcup \emptyset \sqcup \dots \sqcup \emptyset$, we obtain $f(X)<f(Y)$. By supermodularity, $f(X\cup Z)-f(X)\le f(Y\cup Z)-f(Y)$, which together with $f(X)<f(Y)$ and symmetry of $f$ implies
\[  f(Y)=f(X\cup Z)\le f(Y\cup Z)=f((D_2\cup D_1)\sqcup (D_1\cup D_2)\sqcup D_3...\sqcup D_n)). \]
Hence, 
\[ f(D_1\sqcup D_2\sqcup D_3\sqcup \dots \sqcup D_n)\le f((D_2\cup D_1) \sqcup (D_1\cup D_2)\sqcup D_3\sqcup\dots \sqcup D_n). \]

Since $G$ is closed under permutations, intersections, and unions, $(D_1\cap D_2)\sqcup (D_1\cap D_2) \sqcup D_3\sqcup\dots \sqcup D_n)\in G$ and $(D_2\cup D_1) \sqcup (D_1\cup D_2)\sqcup D_3\sqcup\dots \sqcup D_n\in G$.
Repeatedly applying this argument yields the conclusion.
\end{proof}

\begin{lemma}\label{lemma:p_tilde_implements_P}
The reduced form of $p^m$ is $P-\varphi \mathbf{1}_{C^{aud}\cup C^{ic}}$.
\end{lemma}
\begin{proof}
Note that for type $t_i\in C^{allo}$ with quantile $q=F(t_i)$, the probability that $i$ obtains an object under allocation rule $p^m$ equals the probability that there are at most $m-1$ others with types above $t_i$. Using the definition of $P$ and Lemma \ref{lemma:c_functions}, this yields
\begin{align*}
\E_{\t_{-i}}[p^m_i(t_i,\t_{-i})]= \sum_{j=0}^{m-1} \binom{n-1}{j} F(t_i)^{n-1-j} (1-F(t_i))^{j} =P(t_i).
\end{align*}

If $t_i\in C^{aud}$, the probability that $i$ obtains an object under allocation rule $p^m$ equals the probability that there are at most $k-1$ others with types above $t_i$. Again using the definition of $P$ and Lemma \ref{lemma:c_functions}, this yields
\begin{align*}
\E_{\t_{-i}}[p^m_i(t_i,\t_{-i})]= \sum_{j=0}^{k-1} \binom{n-1}{j} F(t_i)^{n-1-j} (1-F(t_i))^{j} = P(t_i)-\varphi.
\end{align*}

Finally, if $t_i\in C^{ic}$ we get
\begin{align*}
\E_{\t_{-i}}[p^m_i(t_i,\t_{-i})]= 0 =P(t_i)-\varphi.
\end{align*}
\end{proof}

\begin{lemma}\label{l:lottery_feasible}
 There is an ex-post allocation rule $p^\ell:\T  \rightarrow \mathbb{R}^n$ allocating the objects not already allocated under $p^m$ that induces the interim rule $\varphi \mathbf{1}_{C^{aud}\cup C^{ic}}$.
\end{lemma}

\begin{proof}
By Lemma \autoref{lemma:border_general}, it is sufficient to check the following inequality:
\begin{equation}\label{eq:new_border}
  \sum_i \int_{E_i} \varphi \mathbf{1}_{C^{aud}\cup C^{ic}} dF(t_i) \le  \int 
  \min\left\{ \sum_{i\in I(\t,E_1,\dots, E_n)}1-p_i^m(\t) , m- \sum_i p^m_i(\t) \right\} \,\mathrm d \F(\t)
\end{equation}

Further, it is sufficient to show this inequality for sets satisfying $E_i\subseteq C^{ic}\cup C^{aud}$ because the right-hand side is increasing in $E_i$ and the left-hand side only depends on $E_i\cap (C^{ic}\cup C^{aud})$.

We now argue that it is sufficient to check the inequality for sets $(E_1,...,E_n)$ where for each $j$ there is $\gamma^j\in[\gamma_2,\gamma_3]$ such that $E_j=[\gamma^j,\gamma_3]$ or $E_j=C^{ic}\cup [\gamma^j,\gamma_3]$.
A simple rewriting shows that  inequality  \eqref{eq:new_border} above is satisfied if the following inequality holds for any $(E_1,...,E_n)$:
\begin{align} \label{eq:long_display}
  &\sum_{j\neq i} \int_{E_j} \varphi \mathbf{1}_{t_j\in C^{aud}\cup C^{ic}} dF(t_j) \nonumber\\
  - & \int_{\T} \min\left\{ \sum_{j\in I(\t,E_1,...,E_n)\setminus \{i\}}1-p_j^m(\t) , m- \sum_{j} p^m_j(\t) \right\} \,\mathrm d \F(\t)\nonumber \\
  + &\int_{E_i} \varphi \mathbf{1}_{t_i\in C^{aud}\cup C^{ic}} dF(t_i)\nonumber\\
  +& \int_{E_i} \int_{\T_{-i}} \min\{ \sum_{j\in I(\t,E_1,...,E_n)\setminus\{i\}}1-p_j^m(\t) , m- \sum_{j} p^m_j(\t) \}\nonumber\\
   & \hspace{1.4cm}- \min\{ \sum_{j\in I(\t,E_1,...,E_n)}1-p_j^m(\t) , m- \sum_j p^m_j(\t) \}   \,\mathrm d \F(\t) \nonumber\\
&\le 0.
\end{align}
The fourth and fifth lines can be written as
\[ -\int_{E_i} \int_{\T_{-i}} \mathbf{1}_{p^m_i(\t)=0 \text{ and } \sum_{j\in I(\t,E_1,...,E_n)}1-p_j^m(\t) \le m- \sum_j p^m_j(\t) }   \,\mathrm d \F(\t). \]
This implies that to maximize the LHS (i) either all or nothing of $C^{ic}$ should be included (because the integrands of the third and fourth line are independent of $t_i$ as $t_i$ varies over $C^{ic}$); (ii) whenever it is optimal to include a type $t_i\in C^{aud}$, then it is optimal to include all higher types $t_i'\in C^{aud}$ (if $p_i^m(t_i,\t_{-i})=1$ then $p_i^m(t_i',\t_{-i})=1$ and the indicator is 0 in either case; if $p_i^m(t_i,\t_{-i})=0$ and $p_i^m(t_i', ,\t_{-i})=1$ then the indicator function changes from $-1$ to $0$ if anything; if $p_i^m(t_i,\t_{-i})=0$ and $p_i^m(t_i' ,\t_{-i})=0$ then the allocation of other agents is not changed and the indicator stays constant). 
We conclude that it is sufficient to check the inequality for sets $(E_1,...,E_n)$ where for each $j$ there is $\gamma^j\in[\gamma_2,\gamma_3]$ such that $E_j=[\gamma^j,\gamma_3]$ or $E_j=C^{ic}\cup [\gamma^j,\gamma_3]$.

We now argue that we can assume, in addition, $E_j=E_i$ for all $i,j$. 
To apply Lemma \ref{lemma:symmetric_solution}, denote the left-hand side of \eqref{eq:long_display} as $f(E_1\sqcup ...\sqcup E_n)$. It can be verified that $f$ is symmetric and supermodular. Also, let $G\subseteq \mathcal{B}(H)$ be the collection of subsets of the form $E_1\sqcup ...\sqcup E_n$, where for each $j$, $E_j=[\gamma^j,\gamma_3]$ or $E_j=C^{ic}\cup [\gamma^j,\gamma_3]$ for some $\gamma^j\in[\gamma_2,\gamma_3]$. Note that $G$ is closed under permutations, intersections, and unions. Lemma \ref{lemma:symmetric_solution} then implies that if $f(E_1\sqcup ...\sqcup E_n)>0$ for $E_1\sqcup ... \sqcup E_n\in G$ then there is $E\sqcup ... \sqcup E\in G$ with $f(E\sqcup ... \sqcup E)>0$.
Hence, it is sufficient to check the inequality \eqref{eq:new_border} for $E_i=E$ for all $i$, where $E=[\gamma,\gamma_3]$ or $E=C^{ic}\cup [\gamma,\gamma_3]$ for some $\gamma\in[\gamma_2,\gamma_3]$. 

Let $p_i^{eff}$ denote the efficient allocation rule,
and $P^{eff}$ denote the induced efficient interim rule. For any $\gamma \in[\gamma_2,\gamma_3]$,
\[n\int_{\gamma}^{\gamma_3} P(t_i) dF_i(t_i) = c^{aud}(\gamma) - c^{allo}(\gamma_3)  \le c^{allo}(\gamma) - c^{allo}(\gamma_3) = n\int_{\gamma}^{\gamma_3} P^{eff}(t_i) dF_i(t_i).\]
Also, 
\[n\int_{C^{ic}} P(t_i) dF_i(t_i) = c^{allo}(0) - c^{allo}(\gamma_1)  = n\int_{C^{ic}} P^{eff}(t_i) dF_i(t_i).\]

Because $p_i^{eff}(\t)\le 1$ for all $i$ and $\sum_{i} p_i^{eff}(\t)\le m$, we get that for any set $E$,
\[ n \int_E P^{eff}(t_i) dF_i(t_i) = \int \sum_{i\in I(\t,E)} p_i^{eff}(\t) d\F(\t)\le \int \min\left\{|I(\t,E)|, m-\sum_{i\not\in I(\t,E)} p_i^{eff}(\t)\right\} d\F(\t). \]

Putting these inequalities together and noting that $p_i^m(\t)\le p_i^{eff}(\t)$ for all $\t$ and $i$, we obtain that for any set $E$ of the form $E=[\gamma,\gamma_3]$ or $E=C^{ic}\cup[\gamma,\gamma_3]$ for $\gamma\in[\gamma_2,\gamma_3]$,
\begin{align*}
n\int_E P(t_i) dF_i(t_i) \le  \int \min\left\{ |I(t,E)|, m-\sum_{j\not\in I(t,E)} p_j^{m}(\t)d\F(\t) \right\}.
\end{align*}

Subtracting $\sum_{j\in I(\t,E)}p_j^m(\t)$ from both sides, we obtain
\begin{align*}
n\int_E \varphi dF_i(t_i) 
\le \int \min\left\{ \sum_{j\in I(t,E)} 1-p_j^m(\t), m-\sum_{j} p^{m}(\t)d\F(\t) \right\}.
\end{align*}
Therefore, the inequality in Lemma \autoref{lemma:border_general}  is satisfied. 
\end{proof}

\subsection{Feasibility of the verification rule}\label{app:verif_feasibility}

To show that $A$ is feasible given $p= p^\ell+ p^m$, we will prove the stronger result that $A$ is feasible given $p^m$.

\begin{proof}[Proof of Lemma \autoref{lemma:border_audit}]     
By Lemma \ref{lemma:border_general} it follows that there is a feasible audit rule $(a_i(\t))_{i\in N}$ inducing $A$ if and only if, for all sets $E=\bigsqcup\limits_{i\in N} E_i$, where each $E_i$ is a Borel subset of $T$, and possible $E_i\neq E_j$, 
\[ \sum_i \int_{E_i} A(t) dF(t) \le \int \min\{|J^p(\t,E)|,k\} d\F(\t). \]     
We can rewrite  this as
\begin{align*}
& \sum_{j\neq i} \int_{E_j} A(t) dF(t) - \int_{\T} \min\left\{\sum_{j\neq i} \1_{t_j\in E_j}p_j(\t),k \right\} d\F(\t) + \int_{E_i} A(t) dF(t) \\
&+ \int_{E_i} \int_{\T_{-i}} \min\left\{ \sum_{j\neq i} \1_{t_j\in E_j}p_j(\t)  , k \right\}- \min\left\{  \sum_{j\neq i} \1_{t_j\in E_j}p_j(\t) +p_i(\t), k \right\}   \,\mathrm d \F(\t) \\ 
&\le 0   
\end{align*}
By an analogous argument as in the proof of Lemma \ref{l:lottery_feasible}, we can conclude that the left-hand side of this expression is supermodular and symmetric as a function of $E$. It therefore satisfies the conditions of Lemma \ref{lemma:symmetric_solution}, and it is enough to check the inequality above for symmetric sets  $E_i=E_j$ for all $i,j$. 
\end{proof}

\begin{lemma}\label{lemma:sufficient_sets}
If 
\begin{align}\label{eq:Border_again}
n\int_{E} A(t_i) \,\mathrm dF(t_i) &\le \int \min\{|J^{p^m}(\t,E)|,k\} \,\mathrm d\F(\t)
\end{align}
holds for all sets of the  form $E= ([\gamma,\gamma_3]\cup [\gamma',1])$ where $\gamma\in[\gamma_1,\gamma_2]$ and $\gamma'\in[\gamma_3,1]$ then it is satisfied for all sets $E\subseteq T$.
\end{lemma}

\begin{proof}
For any collection $(E_1,...,E_n)$, where each $E_i\subseteq T$, define the function $H$ by
\begin{align*}
    H(E_1,...,E_n) &:= \sum_j\int_{E_j}\left(\int_{\T_{-j}}\1_{t_j\in C^{allo}}[p^m_j(t_j,\t_{-j}) -\varphi]+\1_{t_j\in C^{aud}} p^m_j(t_j,\t_{-j}) \right)  \,\mathrm d\F(\t) \\
    &- \int_{\T} \min\left\{k,\sum_j \1_{t_j\in E_j} p^m_j(t_j,\t_{-j})\right\} \,\mathrm d\F(\t).
\end{align*}
Note that, for any $E\subseteq T$, $H(E,...,E)$ equals the left-hand side of \eqref{eq:Border_again} minus the right-hand side of \eqref{eq:Border_again}. Therefore, it is sufficient to show $H(E,...,E)\le 0$.

Fixing an arbitrary agent $i$, we can rewrite $H(E_1,...,E_n)$ as
\begin{align*} 
& \sum_{j\neq i} \int_{E_j} \Big(\int_{\T_{-j}}  \1_{t_j\in C^{allo}} [p^m_j(t_j,\t_{-j})- \varphi] + \1_{t_j\in C^{aud}}p^m_j(t_j,\t_{-j})\Big)d\F(\t) \\
&-   \int_{\T}  \min\big\{k,\sum_{j\neq i} \1_{t_j\in E_j} p^m_j(t_j,\t_{-j})\big\}  d\F(\t)  \\
& + \int_{E_i} \Big(\int_{\T_{-i}}  \1_{t_i\in C^{allo}} [p^m_i(t_i,\t_{-i})- \varphi] + \1_{t_i\in C^{aud}}p^m_i(t_i,\t_{-i})\Big) d\F(\t)  \\
& + \int_{E_i} \Big( \int_{\T_{-i}} \min\big\{k,\sum_{j\neq i} \1_{t_j\in E_j} p^m_j(t_j,\t_{-j})\big\} -  \min\big\{k,\sum_{j\neq i} \1_{t_j\in E_j} p^m_j(t_j,\t_{-j})+p^m_i(t_i,\t_{-i}\big\}   \Big) d\F(\t) 
\end{align*}
Observe that only the third and fourth lines of the left-hand side depend on $E_i$. It can be shown that, for any $E_1,...,E_{i-1}, E_{i+1},...,E_n$, $H(E_1,...,E_{i-1},\cdot, E_{i+1},...,E_n)$ is maximized by setting\footnote{This follows because the integrand in the fourth line equals 0 or -1, and -1 exactly if $\sum_{j\neq i} \1_{t_j\in E_j}p^m_j(t_i,\t_{-i})\ge k$ and $p_i^m(\t)=1$.}
\begin{align}\label{eq:optimal_set}
E_i=C^{aud}\cup \left\{t_i\in C^{allo}: \int_{\left\{\t_{-i}\in \T_{-i}:\sum_{j\neq i} \1_{t_j\in E_j}p^m_j(t_i,\t_{-i})\ge k \right\}} p^m_i(t_i,\t_{-i}) d\F_{-i}(\t_{-i})\ge\varphi  \right\}.    
\end{align}
We are now going to argue that it is maximized by a set of the form  $E_i= ([\gamma,\gamma_3]\cup [\gamma',1])$ where $\gamma\in[\gamma_1,\gamma_2]$ and $\gamma'\in[\gamma_3,1]$. 

Fix arbitrary $\t_{-i}$ and $t_i\in C^{allo}$ such that $t_i< \gamma_2$ and $p^m_i(t_i,\t_{-i})=1$. Choose arbitrary $t_i'\in C^{allo}$ such that $\gamma_2\ge t_i'>t_i$ and  $t_i'\neq t_j$ for all $j\neq i$. Then $p^m_i(t_i',\t_{-i})=1$ and $p^m_j(t_i,\t_{-i})=p^m_j(t_i',\t_{-i})$ for all $j\neq i$ since $t_i$ is among the $m$ highest types and $j$ with type $t_j$ gets an object if and only if either $t_j$ is among the $m$ highest types and in $C^{allo}$ or if $t_j$ is among the $k$ highest types and in $C^{aud}$ (in which case he remains among the $k$ highest types if $i$'s type increases to $t_i$ since $t_i'\le \gamma_2$). Therefore, we conclude that the integral in \eqref{eq:optimal_set} is non-decreasing in $t_i$ on $[\gamma_1,\gamma_2]$. 

Similarly, fix arbitrary $\t_{-i}$ and $t_i\in C^{allo}$ such that $t_i\ge \gamma_3$ and $p^m_i(t_i,\t_{-i})=1$. Choose arbitrary $t_i'\in C^{allo}$ such that $t_i'>t_i$ and $t_i'\neq t_j$ for all $j\neq i$. Then $p^m_i(t_i',\t_{-i})=1$ and $p^m_j(t_i,\t_{-i})=p^m_j(t_i',\t_{-i})$ for all $j\neq i$ since $i$ is among the $m$ highest types and $j$ gets an object if and only if he is either among the $m$ highest types and has a type in $C^{allo}$ or he is among the $k$ highest types and has a type in $C^{aud}$ (in which case he remains among the $k$ highest types since $\gamma_3\le t_i<t_i'$). We conclude that the integral in \eqref{eq:optimal_set} is non-decreasing in $t_i$ on $[\gamma_3,1]$. 
It follows that  $H(E_1,...,E_{i-1},\cdot, E_{i+1},...,E_n)$ is maximized by a set of the form $E_i= ([\gamma,\gamma_3]\cup [\gamma',1])$ where $\gamma\in[\gamma_1,\gamma_2]$ and $\gamma'\in[\gamma_3,1]$. 

Since $H$ is symmetric and supermodular, Lemma \ref{lemma:symmetric_solution} implies that there are $\gamma\in[\gamma_1,\gamma_2]$ and $\gamma'\in[\gamma_3,1]$ such that $H$ is maximized by $(E',...,E')$, where $E'=([\gamma,\gamma_3]\cup [\gamma',1])$. Since the left-hand side of \eqref{eq:Border_again} minus the right-hand side of that inequality evaluated at $E'$ is nonpositive by assumption and equal to $H(E',...,E')$ by construction, we conclude that, for any $E\subseteq T$, $H(E,...,E)\le 0$.
\end{proof}

\begin{lemma}\label{lemma:B2}
Fix arbitrary $\gamma\in [\gamma_1,\gamma_2]$ and $\gamma'\in [\gamma_3,1]$, and let $E= [\gamma,\gamma_3]\cup[\gamma',1]$. Then,
\begin{align*}
n \int_{E} A(t_i)\ dF(t_i)  \le\int \min\{|J^{p^m}(\t,E)|,k\}\ d\F(\t).
\end{align*}
\end{lemma}

\newcommand{\pk}{p^k}
\begin{proof}
Define the allocation rule $\pk:\T\rightarrow [0,1]$ as follows: $\pk_j(\t)=1$ if $t_k\neq t_l$ for all $k\neq l$, $t_j\in C^{allo}\cup C^{aud}$ and $j$'s type is among the $k$ highest ones at $\t$; $\pk_j(\t)=0$ otherwise. It follows from the definition of the merit-with-guarantee rule $P$ in Definition \ref{reduced:opt_mech} and Lemma \ref{lemma:intervals} that
 \begin{align*}
\int_\gamma^{\gamma_2}P(t_i)\ dF(t_i) &= -\frac{1}{n} \int_\gamma^{\gamma_2} \frac{d}{dt_i}( c^{allo}(F(t_i))) \ dt_i  = -\frac{1}{n}[c^{allo}(F(\gamma_2))- c^{allo}(F(\gamma)) ]\\
&\le-\frac{1}{n}[c_\varphi^{aud}(F(\gamma_2))- c_\varphi^{aud}(F(\gamma)) ] = -\frac{1}{n} \int_\gamma^{\gamma_2} \frac{d}{dt_i}(c_\varphi^{aud}(F(t_i))) \ dt_i  \\ 
&=\int_\gamma^{\gamma_2} \left( \sum_{j=0}^{k-1} \binom{n-1}{j} F(t_i)^{n-1-j} (1-F(t_i))^{j} +\varphi\right) dF(t_i)\\
&= \int_\gamma^{\gamma_2}\Big( \int_{\T_{-i}} \pk_i(t_i,\t_{-i})d\F_{-i}(\t_{-i}) + \varphi  \Big)  dF(t_i).
\end{align*}
The equality in the third line follows from Lemma \ref{lemma:c_functions}. 

Similarly, using the definition of $P$, Lemma \ref{lemma:intervals} and Lemma \ref{lemma:c_functions} it follows that,
\begin{align*}
\int_{\gamma_2}^{\gamma_3}P(t_i) \ dF(t_i)&=-\frac{1}{n} \int_{\gamma_2}^{\gamma_3}  \frac{d}{dt_i}(c_\varphi^{aud}(F(t_i))) \ dt_i  \\
&= \int_{\gamma_2}^{\gamma_3} \left( \sum_{j=0}^{k-1} \binom{n-1}{j} F(t_i)^{n-1-j} (1-F(t_i))^{j} +\varphi\right) dF(t_i)\\
&=\int_{\gamma_2}^{\gamma_3} \Big(\int_{\T_{-i}} \pk_i(t_i,\t_{-i}) d\F_{-i}(\t_{-i}) + \varphi \Big) dF(t_i).
\end{align*}
 
Using the definition of $P$ and Lemma \ref{lemma:intervals}, we obtain
\begin{align*}
\int_{\gamma'}^1 P(t_i)\  dF(t_i) &= -\frac{1}{n} \int_{\gamma'}^1 \frac{d}{dt_i}(c^{allo}(F(t_i))) \ dt_i = -\frac{1}{n} [c^{allo}(1)-c^{allo}(F(\gamma'))]  \\
& \le  -\frac{1}{n} \int_{\gamma'}^1 \frac{d}{dt_i}(c_\varphi^{aud} (F(t_i))) \ dt_i = \int_{\gamma'}^1 \Big(\int_{\T_{-i}} \pk_i(t_i,\t_{-i}) d\F_{-i}(\t_{-i}) +\varphi \Big) dF(t_i).
\end{align*}
The inequality follows, since $c^{allo}(F(\gamma'))\le c_\varphi^{aud}(F(\gamma'))$ and $c^{allo}(1)=c_\varphi^{aud}(1)$, and the last equality follows from Lemma \ref{lemma:c_functions}.

Hence,
\begin{align*}
&n \int_{E} A(t_i)\ dF(t_i) =n \int_{E} P(t_i)-\varphi\ dF(t_i)\\
\le& \sum_i \int_{E} \int_{\T_{-i}} \pk_i(t_i,\t_{-i})\ d\F_{-i}(\t_{-i})   \ dF(t_i)\\
=&\sum_i \int_{E} \int_{\T_{-i}} 1_{i\in J^{\pk}(t_i,\t_{-i},E)}\ d\F_{-i}(\t_{-i})   \ dF(t_i)\\
=& \int_{\T} \sum_i 1_{i\in J^{\pk}(t_i,\t_{-i},E)}\ d\F(t)\\
=& \int_{\T} |J^{\pk}(t_i,\t_{-i},E)|\ d\F(t)\\
\le& \int_{\T} \min\left\{k,|J^{p^m}(t_i,\t_{-i},E)|\right\}\ d\F(t)\\
=& \int_{\T} \min\{|J^{p^m}(t_i,\t_{-i},E)|,k\}\ d\F(t),
    \end{align*}
where the second inequality uses that $|J^{p^k}(\t,E)|\le k$ and $p^k_j(\t)\le p^m_j(\t)$  for each $\t\in \T$ and $j\in N$.    
\end{proof}

\begin{lemma}\label{l:A_implementable}
 The interim audit rule $A$ is feasible given $p^m$. 
 \end{lemma}
\begin{proof}
Note that if $A$ is feasible given $p^m$, it is also feasible given $p$ because, for any $E\subset [0,1]$ and $\t\in \T$, if $i\in J^{p^m}(\t,E)$ then $i\in J^{p}(\t,E)$ and therefore 
\[\int \min\{|J^{p^m}(\t,E)|,k\} d\F(\t)\leq \int \min\{|(J^{p}(\t,E)|,k\} d\F(\t).\] 
Moreover, if $a$ is feasible given $p^m$ then it is also feasible given $p$, since $p^m_i(t_i,\t_{-i},s)\leq p_i(t_i,\t_{-i},s)$.  Thus, it suffices to show that the verification rule  $A(t_i)= P(t_i) -\varphi$ is feasible given $p^m$. This follows from Lemmas \ref{lemma:border_audit}, \ref{lemma:sufficient_sets}, and \ref{lemma:B2}.
\end{proof}

\subsection{Proof of Theorem \ref{thm:opt}}\label{app:thm_opt}
By Lemma \ref{l:lottery_feasible}, there is an ex-post lottery rule that allocates the objects not already allocated under $p^m$ such that the interim rule is $\varphi \mathbf{1}_{C^{ic}\cup C^{aud}}$. Denote the ex-post probability of winning an object in either the merit-based allocation or the lottery allocation by $p$. It follows from Lemma \ref{lemma:p_tilde_implements_P} that the interim probability of $p$ is given by $P$. By Lemma \ref{l:A_implementable}, there is an ex-post audit rule $a$ with interim rule $A=P-\varphi$ that is feasible given $p$. It follows that $(P,A)$ is feasible for problem \eqref{P_int}. Since $P$ is optimal for the relaxed problem \ref{relaxed_problem}  and the optimal objective value for the relaxed problem is weakly higher than the optimal value for problem \eqref{P_int} (given the optimal parameter $\varphi$ in the relaxed problem), we conclude that $(P,A)$ are optimal for problem \eqref{P_int}.
\subsection{Proof of Proposition \ref{prop:opt_phi}}\label{app:prop_opt}
The principal's payoff per agent from using a merit-with-guarantee rule with a parameter $\varphi$ is,
\begin{align*}
U(\varphi) &= \int_0^{\gamma_1(\varphi)} -\frac{1}{n} c_{\varphi}^{ic'}(F(t)) t \,\mathrm dF(t) + \int_{\gamma_1(\varphi)}^{\gamma_2(\varphi)} -\frac{1}{n} c^{allo'}(F(t)) t \,\mathrm dF(t) \\
& +\int_{\gamma_2(\varphi)}^{\gamma_3(\varphi)} -\frac{1}{n} c_{\varphi}^{aud'}(F(t)) t \,\mathrm dF(t)+ \int_{\gamma_3(\varphi)}^{1} -\frac{1}{n} c^{allo'}(F(t)) t \,\mathrm dF(t)\\
&= \int_0^{F(\gamma_1(\varphi))}  -\frac{1}{n} c_{\varphi}^{ic'}(q)  F^{-1}(q) \, \mathrm dq+   \int_{F(\gamma_1(\varphi))}^{F(\gamma_2(\varphi))} 
-\frac{1}{n} c^{allo'}  (q) F^{-1}(q) \,\mathrm dq  \\
& +\int_{F(\gamma_2(\varphi))}^{F(\gamma_3(\varphi))} -\frac{1}{n} c_{\varphi}^{aud'}(q) F^{-1}(q) \,\mathrm dq+ \int_{F(\gamma_3(\varphi))}^{1} -\frac{1}{n} c^{allo'}(q) F^{-1}(q) \,\mathrm dq
\end{align*}
The last equality follows from a change of variables with $q=F(t)$. Now define three implicit functions:
\begin{align*}
&h_1(F(\gamma_1(\varphi)),\varphi):= c^{allo}(F(\gamma_1(\varphi)))- c^{ic}_{\varphi}(F(\gamma_1(\varphi)))\\
& h_2(F(\gamma_2(\varphi)),\varphi):= c^{aud}_{\varphi}(F(\gamma_2(\varphi)))- c^{allo}(F(\gamma_2(\varphi)))\\
&h_3(F(\gamma_3(\varphi)),\varphi):= c^{allo}_{\varphi}(F(\gamma_3(\varphi)))- c^{aud}_{\varphi}(F(\gamma_3(\varphi)))
\end{align*}
Note that $h_i(\gamma_i(\varphi),\varphi)=0$ for $i=1,2,3$. 
By the implicit function theorem, we have
\footnotesize
\begin{align*}
\gamma_1'(\varphi)&= -\frac{\frac{\partial(c^{allo}(F(\gamma_1(\varphi)))- c^{ic}_{\varphi}(F(\gamma_1(\varphi))))}{\partial \varphi}}{[c^{allo'}(F(\gamma_1(\varphi)))- c_{\varphi}^{ic'}(F(\gamma_1(\varphi)))]f(\gamma_1(\varphi))}=  - \frac{nF(\gamma_1(\varphi))}{[c^{allo'}(F(\gamma_1(\varphi)))- c_{\varphi}^{ic'}(F(\gamma_1(\varphi)))]f(\gamma_1(\varphi))} \\
\gamma_2'(\varphi)&= -\frac{\frac{\partial(c^{aud}_{\varphi}(F(\gamma_2(\varphi)))- c^{allo}(F(\gamma_2(\varphi))))}{\partial \varphi}}{[c_{\varphi}^{aud'}(F(\gamma_2(\varphi)))- c^{allo'}(F(\gamma_2(\varphi)))]f(\gamma_2 \varphi))}= - \frac{n(1-F(\gamma_2(\varphi)))}{[c_{\varphi}^{aud'}(F(\gamma_2(\varphi)))- c^{allo'}(F(\gamma_2(\varphi)))]f(\gamma_2 (\varphi))}\\
\gamma_3'(\varphi)&=- \frac{\frac{\partial(c^{allo}(F(\gamma_3(\varphi)))- c_{\varphi}^{aud}(F(\gamma_3(\varphi))))}{\partial \varphi}}{[c^{allo'}(F(\gamma_3(\varphi)))- c_{\varphi}^{aud'} (F(\gamma_3(\varphi)))]f(\gamma_3(\varphi))}= - \frac{-n(1-F(\gamma_3(\varphi))}
{[c^{allo'}(F(\gamma_3(\varphi)))- c_{\varphi}^{aud'} (F(\gamma_3(\varphi)))]f(\gamma_3(\varphi))}
\end{align*}
\normalsize
Now let us compute  $U'(\varphi)$, using Leibniz integral rule, 
\begin{align*}
U'(\varphi)&= -\frac{1}{n}c_{\varphi}^{ic'}(F(\gamma_1(\varphi)))\gamma_1(\varphi)f(\gamma_1(\varphi))\gamma_1'(\varphi) + \int_0^{F(\gamma_1(\varphi))} -\frac{1}{n}\frac{\partial}{\partial \varphi}c_{\varphi}^{ic'}(F(\gamma_1(\varphi))) F^{-1}(q) \, \mathrm dq \\
&-\frac{1}{n} c^{allo'}(F(\gamma_2(\varphi)))\gamma_2(\varphi)f(\gamma_2(\varphi))\gamma_2'(\varphi)+ \frac{1}{n} c^{allo'}(F(\gamma_1(\varphi)))\gamma_1(\varphi)f(\gamma_1(\varphi))\gamma_1'(\varphi)\\
&- \frac{1}{n} c_{\varphi}^{aud'}(F(\gamma_3(\varphi)))\gamma_3(\varphi)  f(\gamma_3(\varphi))\gamma_3'(\varphi)+ \frac{1}{n} c_{\varphi}^{aud'}(F(\gamma_2(\varphi)))\gamma_2(\varphi)  f(\gamma_2(\varphi))\gamma_2'(\varphi)\\
&- \int_{F(\gamma_2(\varphi))}^{F(\gamma_3(\varphi))} \frac{1}{n} \frac{\partial}{\partial \varphi} c_{\varphi}^{aud'}(q) F^{-1}(q) \, \mathrm dq  +\frac{1}{n}c^{allo'}(F(\gamma_3(\varphi)))\gamma_3(\varphi) f(\gamma_3(\varphi))\gamma_3'(\varphi)\\
&= \int_0^{F(\gamma_1(\varphi))}  F^{-1}(q) \, \mathrm dq + \int_{F(\gamma_2(\varphi))} ^{F(\gamma_3(\varphi))}  F^{-1}(q) \, \mathrm dq\\ 
&+\gamma_1(\varphi)f(\gamma_1(\varphi))\gamma_1'(\varphi)\frac{1}{n}\Big(  c^{allo'}(F(\gamma_1(\varphi))-c_{\varphi}^{ic'}(F(\gamma_1(\varphi))\Big)  \\
&+\gamma_2(\varphi)f(\gamma_2(\varphi))\gamma_2'(\varphi)\frac{1}{n}\Big(c_{\varphi}^{aud'}(F(\gamma_2(\varphi)))- c^{allo'}(F(\gamma_2(\varphi)))\Big)\\
&+\gamma_3(\varphi)f(\gamma_3(\varphi))\gamma_3'(\varphi)\frac{1}{n}[c^{allo'}(F(\gamma_3(\varphi)))- c_{\varphi}^{aud'}(F(\gamma_3(\varphi)))]\\
& = -\gamma_1(\varphi)F(\gamma_1(\varphi)) - \gamma_2(\varphi)[1-F(\gamma_2(\varphi)) ] +\gamma_3(\varphi)[1-F(\gamma_3(\varphi)) ]  \\
&+\int_{0}^{\gamma_1(\varphi)}  t \, \mathrm dF(t) +\int_{\gamma_2(\varphi)}^{\gamma_3(\varphi)}  t \, \mathrm dF(t).
\end{align*}
Rewriting this we obtain:
\[
\gamma_1(\varphi)F(\gamma_1(\varphi)) +\gamma_2(\varphi)[1-F(\gamma_2(\varphi)) ] +\gamma_3(\varphi)[1-F(\gamma_3(\varphi)) ]  =\int_{0}^{\gamma_1(\varphi)}  t \, \mathrm dF(t) +\int_{\gamma_2(\varphi)}^{\gamma_3(\varphi)}  t \, \mathrm dF(t),
\]
for having $U'(\varphi)=0$, just as stated in Proposition \ref{prop:opt_phi}. 

\end{appendix}

\newpage
\begin{small}
	\bibliography{references}
\end{small}

\end{document}